\documentclass[10pt,journal,compsoc]{IEEEtran}
\usepackage[T1]{fontenc}
\ifCLASSOPTIONcompsoc

\usepackage{cite}    
\usepackage{bbm}
\usepackage{bm}
\usepackage{amsmath}
\usepackage{amsthm}
\usepackage{float}
\usepackage{setspace}
\usepackage{amssymb}
\usepackage{stfloats}
\usepackage{cite}
\usepackage{ragged2e}
\usepackage[ruled,vlined,linesnumbered]{algorithm2e}
\usepackage{amsfonts}
\usepackage{mathrsfs}
\usepackage{amsmath,amsthm}
\usepackage{array,booktabs}
\usepackage{subfigure}
\usepackage{multirow}
\usepackage{cuted}
\usepackage{multicol}
\usepackage{graphicx}
\usepackage{subfigure}
\usepackage{graphicx,xcolor,bm}
\usepackage{hyperref}
\usepackage{threeparttable}
\usepackage{dcolumn}
\usepackage{setspace}
\usepackage{makecell}
\usepackage{lipsum}
\usepackage{enumerate}
\usepackage{graphicx} 		
\usepackage{subfigure}		
\usepackage{multirow}    
\usepackage{amssymb}     
\usepackage{xcolor}       
\usepackage{booktabs}    
\usepackage{tabularx}    
\usepackage{multirow}
\usepackage{array}
\usepackage{amssymb}
\usepackage{xcolor}
\usepackage{booktabs}
\usepackage{pifont}

\usepackage{booktabs}
\usepackage[table]{xcolor}

\newtheorem*{myPro}{Proof}

\newtheorem{thm}{Theorem}

\usepackage{cite,bm}
\graphicspath{{figures/}}
\begin{document}
	\title{Look One Step Ahead: Forward-Looking Incentive Design with Strategic Privacy for Proactive Service Provisioning over Air-Ground Integrated Edge Networks}
	
	\author{Sicheng Wu, Minghui Liwang, \IEEEmembership{Senior Member}, \IEEEmembership{IEEE}, Yangyang Gao, \\
		Deqing Wang, \IEEEmembership{Member}, \IEEEmembership{IEEE}, Wenbo Zhu, \IEEEmembership{Member}, \IEEEmembership{IEEE}, Yiguang Hong, \IEEEmembership{Fellow}, \IEEEmembership{IEEE}, \\ Wei Ni, \IEEEmembership{Fellow}, \IEEEmembership{IEEE}, and Seyyedali Hosseinalipour, \IEEEmembership{Senior Member}, \IEEEmembership{IEEE}
		\vspace{-0.2cm}
		\thanks{S. Wu (wusch2025@lzu.edu.cn) is with School of Information Science and Engineering, Lanzhou University, Gansu, China. M. Liwang (minghuiliwang@tongji.edu.cn), Y. Gao (gaoyangyang2004@outlook.com), W. Zhu (wbzhu@tongji.edu.cn) and Y. Hong (yghong@tongji.edu.cn) are with the Department of Control Science and Engineering, the National Key Laboratory of Autonomous Intelligent Unmanned Systems, and also with Frontiers Science Center for Intelligent Autonomous Systems, Ministry of Education, Tongji University, Shanghai, China. D. Wang (deqing@xmu.edu.cn) is with the School of Informatics, Xiamen University, Fujian, China. W. Ni (Wei.Ni@ieee.org) is with the School of Engineering, Edith Cowan University, Perth, WA 6027, and the School of Computer Science and Engineering, University of New South Wales, Sydney, NSW, Australia. S. Hosseinalipour (alipour@buffalo.edu) is with Department of Electrical Engineering, University at Buffalo-SUNY, USA.
	}}

	\IEEEtitleabstractindextext{
		\begin{abstract}
			\justifying
			In \underline{a}ir-\underline{g}round \underline{i}ntegrated \underline{n}etworks (AGINs), unmanned aerial vehicles (UAVs) are employed to provide on-demand edge services to ground vehicles. Realizing this vision, however, requires carefully designed incentives to coordinate interactions among self-interested/selfish and mobile participants. This challenge is further exacerbated by the inherently dynamic nature of AGINs, where rapid spatio-temporal variations introduce significant uncertainty in matching service providers (i.e., UAVs) and requesters (i.e., vehicles). At the same time, existing real-time service provisioning methods typically relies on precise vehicle trajectory information to achieve accurate UAV-to-vehicle matching, which not only raises privacy concerns but also incurs decision latency. To address these intertwined challenges, we propose a novel framework, termed \underline{l}ook \underline{o}ne-\underline{s}tep \underline{a}head (LOSA), for efficient and privacy-aware service provisioning between UAVs and vehicles. The key insight is to exploit the predictable vehicle travel time between road intersections, which enables a principled decomposition of the service provisioning process into two tightly coupled phases: \textit{(i)} a privacy-aware \textit{look-ahead phase} and \textit{(ii)} a lightweight real-time \textit{execution phase}. The look-ahead phase allows vehicles to adaptively adjust their privacy budgets based on historical utility, thereby balancing trajectory exposure and vehicle-to-UAV matching accuracy in a controlled manner. Leveraging the resulting information, a carefully designed double auction mechanism establishes binding \textit{one-step-ahead agreements (OSAAs)} through trajectory similarity clustering, while simultaneously constructing \textit{preference lists} to hedge against potential disruptions caused by vehicle mobility uncertainty. With these agreements in place, the execution phase enforces the pre-established OSAAs and associated preference lists, while resolving potential real-time resource conflicts \textit{without} requiring costly re-negotiations. This two-phase design not only reduces real-time computational overhead but also preserves robustness against dynamic network conditions. We analytically corroborate that LOSA guarantees key economic properties, including \textit{truthfulness}, \textit{individual rationality}, and \textit{budget balance}. Extensive experiments on real-world datasets (DAIR-V2X, HighD, and RCooper) demonstrate that LOSA achieves superior privacy protection while lowering transaction latency, thereby enabling robust, privacy-preserving, and time-efficient service provisioning in dynamic AGIN environments, compared to baseline approaches.
		\end{abstract}

	}
	
\maketitle
\IEEEdisplaynontitleabstractindextext
%
\IEEEpeerreviewmaketitle

\section{Introduction}
\label{sec:intro}

\IEEEPARstart{T}{he} rapid advancement of intelligent vehicles and the low-altitude economy have catalyzed the development of \underline{a}ir-\underline{g}round \underline{i}ntegrated \underline{n}etworks (AGINs). In AGINs, unmanned aerial vehicles (UAVs), owing to their high mobility, flexible deployment \cite{11014565}, and line-of-sight links \cite{10843374}, can act as mobile edge servers to assist ground vehicles in processing computation/data-intensive tasks \cite{10413524, 11203884, 11030943}, such as autonomous driving and real-time navigation. To facilitate edge service provisioning, double auction mechanisms are widely adopted due to their effectiveness in coordinating self-interested participants \cite{10342859, 11204540}. However, they face a fundamental tension: high efficiency requires fine-grained vehicle trajectories, conflicting with location privacy \cite{10937134}. Privacy-preserving perturbation mitigates exposure but degrades spatial accuracy, leading to infeasible UAV–vehicle matching or suboptimal auction outcomes \cite{10430221}. Moreover, existing designs largely rely on purely real-time trading, ignoring the partially predictable nature of vehicle mobility, which incurs unnecessary latency and underutilizes temporal information for service coordination \cite{11183623, 11155207}.


These limitations motivate a service provisioning framework exploiting mobility predictability while preserving trajectory privacy over AGINs. Thus, we propose the new \underline{l}ook \underline{o}ne-\underline{s}tep \underline{a}head (LOSA) framework, which predicts each vehicle’s arrival at its next intersection instead of requiring full trajectory disclosure. This coarse-grained abstraction enables privacy-sensitive and computation-intensive auction decisions to be shifted to a \textit{look-ahead phase} based on intersection-level mobility alignment. Subsequently, LOSA decouples decision-making from execution: vehicles and UAVs are matched via a privacy-aware double auction on predicted intersection arrivals, and resulting agreements are enforced with minimal real-time computation during the \textit{execution phase}. By replacing repeated online optimization with one-step mobility prediction, LOSA reduces latency while limiting trajectory exposure.

\subsection{Challenges and Motivation}
Designing a privacy-preserving and time-efficient auction-based service provisioning mechanism for AGINs is a non-trivial task, as it requires addressing several fundamental challenges. To holistically tackle these issues, we articulate them as a set of key research questions in the following: 

\noindent$\bullet$ \textit{How to reconcile the conflict between spatio-temporal obfuscation and matching accuracy?} 
High-quality service matching in AGINs relies on accurate spatio-temporal alignment between service buyers and sellers. However, directly sharing mobility information exposes vehicles to privacy risks. This creates an inherent trade-off: excessive information obfuscation degrades spatial fidelity and can lead to infeasible or inefficient matches, e.g., assigning a UAV that is physically out of range to a vehicle, whereas insufficient obfuscation preserves matching accuracy but reveals sensitive mobility patterns. To address this, we develop a location obfuscation framework based on discretized polar coordinates, coupled with a dynamic privacy budget mechanism. In this framework, instead of relying on static noise injection, each vehicle adaptively adjusts its privacy budget (i.e., $\xi_t^n$) according to historical service utility. Privacy is selectively relaxed only when it yields tangible utility gains, enabling accurate matching while protecting location information against potential eavesdroppers.

\noindent$\bullet$ \textit{How to overcome the service latency and uncertainty caused by high mobility?}
In AGINs, the strong spatio-temporal coupling between UAVs and vehicles causes service trading opportunities to evolve rapidly as participants move, making timely service coordination challenging. At the same time, service demand is inherently stochastic and context-dependent. For example, a vehicle's service request may arise probabilistically, rendering static service allocation strategies ineffective in the presence of sudden demand fluctuations or unexpected seller (i.e., UAV) unavailability. To address these challenges, we propose to proactively utilize idle travel time by introducing a look-ahead phase (Phase 1 of LOSA), during which one-step-ahead agreements (OSAAs) are established based on predicted intersection-level mobility. This anticipatory design shifts decision-making away from latency-critical moments and enables timely service coordination. To account for stochastic demand and potential disruptions, we construct buyer-specific preference lists that serve as contingency plans for the execution phase (Phase 2 of LOSA). These pre-computed alternatives allow the system to perform fast/real-time supplementary UAV-to-vehicle matching when needed, without incurring the overhead of full auction negotiations.

\noindent$\bullet$ \textit{How to ensure economic properties in a market with endogenous privacy costs?}
Incorporating privacy loss into the utility function of resource buyers/sellers fundamentally complicates the economic design of the service provisioning market. In particular, rational buyers will participate in the market only if their expected service value exceeds both the monetary payment and the associated privacy cost. On the other hand, the service provisioning mechanism must remain robust to strategic behaviors of participants, ensuring that they cannot benefit from misreporting their private valuations or privacy sensitivities, even under stochastic demand conditions. To address this challenge, we develop a Vickrey–Clarke–Groves (VCG)-based pricing mechanism that incorporates privacy costs meticulously into the utility formulation. To ensure scalability in dynamic AGIN settings, this pricing scheme is further integrated with path similarity clustering, which reduces the computational complexity of the matching process while preserving the matching efficiency. Building on this design, we rigorously establish that our proposed service provisioning mechanism satisfies key economic properties, including truthfulness, budget balance, and individual rationality, even when privacy costs are endogenously embedded in the utility function.

\subsection{Novelty and Contribution}
Building on the components developed in response to the aforementioned formulated research questions, our core contributions can be summarized as follows:

\noindent$\bullet$ \textit{Framework design: a novel ``look one-step ahead'' service provisioning framework over AGINs.} 
We introduce a temporal abstraction to alleviate the spatio-temporal coupling induced by the mobility of vehicles and UAVs. Specifically, by leveraging the otherwise underutilized travel time between intersections, LOSA decomposes the service provisioning process into two coordinated phases: a privacy-aware look-ahead phase (Phase 1) and a real-time execution one (Phase 2). This decomposition allows OSAAs to be established before execution, converting mobility uncertainty into structured anticipatory decisions. Consequently, LOSA stabilizes real-time service provisioning and reduces reliance on reactive, latency-sensitive coordination in the network.

\noindent$\bullet$ \textit{Micro-perspective: achieving dynamic privacy-utility trade-off.} 
In LOSA, we address the inherent tension between UAV-to-vehicle matching accuracy and user privacy. To this end, rather than relying on static noise injection, we develop a dynamic privacy-budget adjustment mechanism integrated with a three-step location obfuscation model that discretizes both the perturbation radius and angle. This design establishes a closed-loop feedback mechanism between historical matching utility and privacy budgets, allowing each vehicle to autonomously regulate its level of location exposure (i.e., $\xi_t^n$). As a result, privacy is selectively relaxed only when it leads to meaningful utility gains, ensuring accurate service matching when needed while maintaining sufficient uncertainty in the shared information to mislead potential adversaries/eavesdroppers.

\noindent$\bullet$ \textit{Solution design: a two-phase synergistic decision with contingency plans.} 
To enable efficient UAV-to-vehicle matching under privacy-preserving inputs/information and demand uncertainty, we develop a two-phase double-auction framework that integrates proactive planning with real-time execution. In Phase 1, we cluster participants based on trajectory similarity using the Fréchet distance and jointly optimize active UAV deployments, which allows us to derive risk-aware OSAAs along with VCG-based pricing within the auction. We also construct buyer-specific preference lists that serve as lightweight contingency plans. Building on these pre-computed decisions, Phase 2 executes the established OSAAs in real-time and, when disruptions occur, seamlessly switches to backup sellers according to the preference lists. This design eliminates the need to re-initiate time-consuming negotiations and ensures service continuity under stochastic demand and mobility uncertainty.

\noindent$\bullet$ \textit{Validation: rigorous theoretical guarantees and real-world data-driven evaluation.} 
We validate our proposed LOSA framework through both theoretical analysis and experiments. Theoretically, we rigorously prove that LOSA satisfies key economic properties, including individual rationality, budget balance, and truthfulness. Empirically, we conduct evaluations using a fused dataset constructed from three real-world sources (DAIR-V2X, HighD, and RCooper), enabling realistic modeling of mobility and service dynamics. Results demonstrate that LOSA consistently surpasses state-of-the-art benchmarks, providing enhancement in privacy protection strength while sustaining a high time-efficiency in dynamic AGINs.

\section{Literature Investigation}
\label{sec:related_work}

We provide an overview of the related literature, while highlighting the key differences of this work.

\noindent $\bullet$ \textit{Service provisioning in AGINs.} The integration of UAVs with terrestrial networks has been recognized as a promising strategy to augment both coverage and edge computing capabilities of ground infrastructures. For instance, \textit{Liu et al.} \cite{11192086} proposed a space-air-ground edge computing architecture leveraging UAVs as edge servers to reduce processing delays. \textit{Li et al.} \cite{9965749} investigated the joint optimization of trajectory and transmit power, using iterative successive convex approximation. \textit{Zhang et al.} \cite{10521811} addressed joint UAV trajectory planning and resource allocation in heterogeneous edge networks through a potential game framework. \textit{Park et al.} \cite{10238831} introduced an attention-driven multi-agent deep reinforcement learning (MADRL) approach for joint UAV trajectory and resource optimization to reduce the energy use and delay of service provisioning. Other studies have considered real-time decision-making under dynamic network edge conditions. For example, \textit{He et al.} \cite{10700794} studied online task scheduling via Lyapunov optimization, and \textit{Xing et al.} \cite{11362958} developed MADRL-based real-time UAV service provisioning strategies. Despite their cotributions, existing approaches are generally reactive, ignoring predictable mobility between intersections, which delays decisions and risks mismatches. LOSA instead uses short-term predictions to establish agreements in advance, reducing the need for online optimization requirements and the risk of UAV-to-vehicle mismatches.

\noindent $\bullet$ \textit{Double auctions for edge service trading.} Double auction mechanisms are widely adopted to facilitate resource sharing. For instance, \textit{Wu et al.} \cite{11224729} investigated iterative double auction for task scheduling over multi-server edge marketplaces to maximize social welfare. Moreover, \textit{Liu et al.} \cite{10892012} proposed a truthful mixed double auction model to facilitate cooperative resource provisioning in vehicle-assisted edge networks. Although such VCG-based auctions have desirable economic properties, directly applying them in large-scale AGINs is computationally expensive due to their combinatorial nature. To improve scalability, some works have proposed lightweight or decomposed auction designs. For instance, \textit{Qiu et al.} \cite{11115158} introduced a relay-based auction mechanism that decomposes the auction process into smaller sub-problems, while \textit{Khamse-Ashari et al.} \cite{9361144} developed an iterative auction game to enable distributed and low-complexity service provisioning. Beyond computational burden, existing mechanisms assume static demand and rely on reactive clearing, making them inadequate for stochastic AGIN environments. Sustaining efficiency thus requires frequent re-clearing, incurring extra latency. Recent works instead explore anticipatory or multi-phase decisions. For example, \textit{Wu et al.} \cite{11204540} proposed a two-phase double auction with overbooking to handle dynamic supply and demand, while \textit{Qi et al.} \cite{11155207} introduced a future resource bank framework that combines offline risk-aware contracts with online backup matching. Building on this line of work, LOSA introduces a mobility-aware look-ahead auction using short-term trajectory predictions. Subsequently, by clustering similar vehicle trajectories for tractable VCG pricing and obtaining OSAAs in advance, it shifts decisions to the travel interval and handles demand uncertainty without repeated market clearing.

\noindent $\bullet$ \textit{Trajectory privacy preservation for location-based services.} Privacy is a critical concern in AGINs, where accurate trajectory information is essential for effective service matching. This creates a fundamental tension between preserving user privacy and maintaining matching accuracy. To address this issue, a wide range of privacy-preserving techniques have been developed, including $k$-anonymity, differential privacy (DP), and Geo-Indistinguishability (Geo-I). For example, \textit{Huang et al.} \cite{11106247} proposed a DP-based framework to protect users' behavioral patterns and trajectories in location-based services, while \textit{Min et al.} \cite{10325612} developed a personalized 3D location privacy scheme based on Geo-I to obfuscate user locations while preserving service quality. Nevertheless, most of these approaches rely on static privacy mechanisms, like fixed DP noise or preset budgets, enforcing a rigid privacy–utility trade-off. Some works have introduced utility-aware designs. Particularly, \textit{Yao et al.} \cite{10197160} balanced privacy and utility, while \textit{Chen et al.} \cite{10891693} optimized the privacy–utility trade-off by quantifying the inherent privacy risk before data publishing. However, they still relied on manual budgets or static privacy thresholds, limiting their adaptation to AGINs' dynamic service conditions. More recently, efforts have been made to integrate privacy considerations into economic service provisioning. For example, \textit{Song et al.} \cite{10964145} incorporated DP into a double-layer auction framework to jointly optimize resource allocation and privacy budgets. \textit{Xu et al.} \cite{11159128} proposed a privacy-preserving auction using combinatorial obfuscation to protect user locations in UAV-assisted services. Compared to these methods that treat privacy budgets as fixed, LOSA makes privacy dynamic and tunable: by linking each user’s budget to historical matching utility, it creates a closed-loop feedback that lets vehicles adjust trajectory exposure, allowing privacy to be relaxed when acceptable and beneficial, achieving an adaptive, context-aware privacy–utility trade-off.


\begin{table}[t!]
	\centering
	\caption{Comparative summary of existing literature and our work's features. (DA: Double Auction, PP: Privacy Preservation, CE: Computational Efficiency, LA: Look-Ahead, DPB: Dynamic Privacy Budget, SD: Stochastic Demand, TP: Trajectory Planning)}
	\vspace{-0.3cm}
	\label{tab:comparison}
	\scriptsize
	\setlength{\tabcolsep}{5pt}
	\renewcommand{\arraystretch}{1.2}
	\begin{tabular}{c|ccc|cccc}
		\toprule
		
		\multirow{2}{*}{\textbf{Reference}} & \multicolumn{3}{c}{\textbf{Basic Attributes}} & \multicolumn{4}{c}{\textbf{Advanced Functionalities}} \\
		\cmidrule(lr){2-4} \cmidrule(lr){5-8}
		& \cellcolor{gray!20}DA & \cellcolor{gray!20}PP & \cellcolor{gray!20}CE & \cellcolor{gray!20}LA & \cellcolor{gray!20}DPB & \cellcolor{gray!20}SD & \cellcolor{gray!20}TP \\
		\midrule
		\cite{11192086}, \cite{9965749}, \cite{10521811}, \cite{10238831} &  &  & $\checkmark$ &  &  &  & $\checkmark$ \\ 
		\cite{10700794}, \cite{11362958} &  &  & $\checkmark$ &  &  & $\checkmark$ &  \\
		\cite{11224729}, \cite{10892012}, \cite{11115158}, \cite{9361144} & $\checkmark$ &  & $\checkmark$ &  &  &  &  \\
		\cite{11204540}, \cite{11155207} & $\checkmark$ &  & $\checkmark$ & $\checkmark$ &  & $\checkmark$ &  \\
		\cite{11106247}, \cite{10325612}, \cite{10197160}, \cite{10891693} &  & $\checkmark$ &  &  &  &  &  \\
		\cite{10964145}, \cite{11159128} & $\checkmark$ & $\checkmark$ & $\checkmark$ &  &  &  & $\checkmark$ \\
		\midrule
		\rowcolor{gray!15}
		\textbf{LOSA (Ours)} & $\checkmark$ & $\checkmark$ & $\checkmark$ & $\checkmark$ & $\checkmark$ & $\checkmark$ & $\checkmark$ \\
		\bottomrule
	\end{tabular}
	\vspace{-0.3cm}
\end{table}

\section{System Model and Operational Overview}

We consider an AGIN architecture, comprising vehicles and low-altitude UAVs. Within this setting, our goal is to enable time-efficient, privacy-preserving, and reliable service provisioning under dynamic and uncertain conditions, while navigating the interplay among the following factors: \textit{(i)} spatio-temporal coupling between service requestors (i.e., vehicles) and providers (i.e., UAVs); \textit{(ii)} uncertainty in service demands and availability; and \textit{(iii)} the trade-off between service quality/effectiveness and user privacy. To address these factors, we next introduce the key components of the system model.

\begin{figure}[htb]
	\centering
	\includegraphics[trim=0cm 0cm 0cm 0cm, clip, width=\columnwidth]{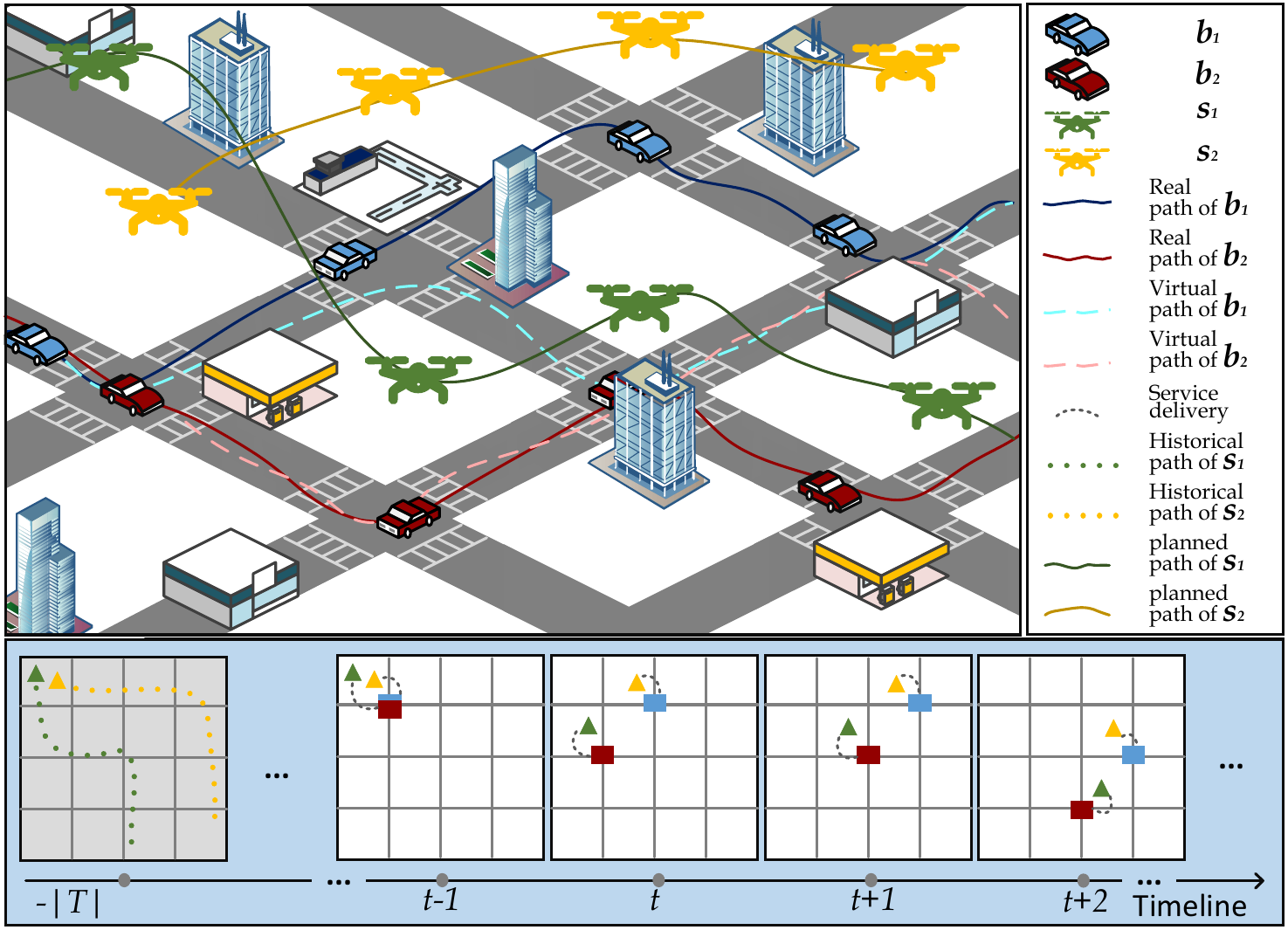}
	\vspace{-0.6cm}
	\caption{AGIN framework of our interest with privacy-preserving vehicles and trajectory-planning enabled UAVs, along with a timeline example.}
	\label{fig_1}
\end{figure}

\vspace{-0.3cm}
\subsection{System Description: Map, Time Horizon, and Participants}

\textit{Map and time discretization.} We consider an orthogonally structured road network that forms a grid pattern, similar to that of Manhattan. At each intersection, a roadside unit (RSU) is deployed as an auctioneer, facilitating local service coordination\footnote{We consider a single logical auctioneer entity whose functions are implemented by distributed units, namely RSUs. Accordingly, in Section \ref{sec:AuctioneerUtility}, we represent the auctioneer’s utility as a single, unified entity.}. To capture the spatio-temporal dynamics of the system, the time horizon is discretized into a set of timeslots $\bm{\mathcal{T}} = \left\{1, \ldots, t, \ldots, {|\bm{\mathcal{T}}|}\right\}$. Each timeslot $t$ has a fixed duration $\Delta t$ defined as the time required for a vehicle to traverse the distance between two adjacent intersections. This discretization aligns mobility with the grid structure and enables tractable modeling of vehicle movement. Under this model, vehicle speeds are represented as integer multiples of the intersection-to-intersection distance per timeslot. If the distance between two adjacent intersections is 100 meters and $\Delta t$ is defined such that traversing one intersection-to-intersection segment takes one timeslot, then a vehicle with unit speed moves from its current intersection to the next intersection in one timeslot. A faster vehicle may traverse two consecutive intersections within the same timeslot, corresponding to a speed of two intersection-lengths per $\Delta t$. This representation ensures that all vehicle positions are aligned with intersection locations at each timeslot, hence simplifying subsequent analyses. 

\noindent\textit{Key participants in the service provisioning market.} Considering two types of participants, namely buyers and sellers, as detailed below:

\noindent$\bullet$\textbf{ Buyers (vehicles):} We denote the set of buyers (i.e., moving vehicles) as $\bm{\mathcal{B}} = \left\{b_1, \ldots, b_n, \ldots, b_{|\bm{\mathcal{B}}|}\right\}$. Each buyer $b_n$ follows a private trajectory $\bm{\mathcal{L}}_{n}^{\mathsf{b}}=\left\{l_{n}^{\mathsf{b},1},\ldots,l_{n}^{\mathsf{b},{|\bm{\mathcal{T}}|}}\right\}$, where $l_{n}^{\mathsf{b},t}$ describes its coordinates at the end of timeslot $t$. These trajectories are considered sensitive information and are not disclosed during service provisioning. Instead, buyers may strategically obfuscate their mobility patterns to preserve privacy. In addition, buyers exhibit heterogeneous and potentially multi-type service demands. For example, a vehicle may request multiple services, such as real-time sensing data for safe navigation, access to nearby charging infrastructure, or context-aware environmental information.

\noindent$\bullet$\textbf{ Sellers (UAVs):} We denote the set of sellers (i.e., UAVs) as $\bm{\mathcal{S}} = \left\{s_1, \ldots, s_m, \ldots, s_{|\bm{\mathcal{S}}|}\right\}$. Each seller $s_m$ is equipped with sensing, processing, and storage capabilities, and maintains a record of its historical flight trajectory $\bm{\mathcal{L}}_{m}^{\mathsf{s}}$. Each seller can provide multiple types of services. We consider $J$ service types, such as high-definition (HD) map updates and real-time obstacle detection. The asking price of a seller is closely tied to the spatio-temporal characteristics of the service, including its location, availability, and quality.

The above-described market operates under a flexible many-to-many (M2M) matching paradigm. Specifically, a \textit{matching} defines the assignment between buyers and sellers, together with the allocation of service types and corresponding resources. Under this definition, a buyer's composite demand may be fulfilled by multiple sellers, each providing a subset of the requested services, while a seller can simultaneously serve multiple buyers by allocating its available resources across different service requests. This formulation captures the combinatorial nature of service provisioning in dynamic AGIN environments.

\subsection{Overview of the LOSA Framework}
To address the coupled challenges of spatio-temporal dynamics and trajectory privacy, we develop LOSA. Unlike conventional approaches that execute a single auction when buyers and sellers meet at intersections, LOSA decomposes the trading process into two coordinated phases that exploit the travel time between intersections to improve both decision efficiency and adaptability. These two phases are described below. We provide an illustrative example remove the involving buyer $b_1$ over timeslots $t-1$ and $t$, as illustrated in Fig. \ref{fig_1}. Details are provided in Appendix \ref{Example}. 

\textit{Phase 1. Look-Ahead Decision Phase.} This phase is executed during timeslot $t$ for services that will be delivered at the next intersection, i.e., at the end of timeslot $t$. Its objective is to establish OSAAs between buyers and sellers, serving as binding commitments for future service transaction. By forming these agreements in advance, LOSA enables coordinated decision-making while accounting for mobility and uncertainty. Below, we describe the behavior of each entity in this phase.

\noindent$\bullet$\textbf{ Buyer behavior:} Based on its true/private trajectory $\bm{\mathcal{L}}_{n}^{\mathsf{b},t}$ and its current privacy budget $\xi_{n}^{t}$, each buyer $b_n$ generates and reports a privacy-preserving \textit{virtual trajectory} $\hat{\bm{\mathcal{L}}}_{n}^{\mathsf{b},t}$, along with its bid $\bm{bid}_{n}^{t}$ derived from its service valuation $\bm{\mathcal{V}}_n$. This mechanism allows the buyer to balance privacy and service quality, as excessive deviation from the true trajectory may degrade the subsequent matching accuracy.

\noindent$\bullet$\textbf{ Seller behavior:} Each seller $s_m$ submits its asking price $\bm{ask}_m^{t}$ based on its service cost $\bm{\mathcal{C}}_m$. In addition to its asking price, its suitability for serving a buyer is determined by its historical trajectory, denoted as $\bm{\mathcal{L}}_{m}^{\mathsf{s},t}$, which reflects its spatial coverage and service capability over time.

\noindent$\bullet$\textbf{ Auctioneer behavior:} The trading process is coordinated by an auctioneer, implemented through RSUs. Based on the buyers' and sellers' reports, the auctioneer evaluates the \textit{trajectory similarity} ($\Gamma_{m,n}^{t}$) between each buyer's virtual trajectory and each seller's historical trajectory.

Based on the submitted bids, asks, and trajectory similarity measures, the auctioneer determines the winning buyer–seller assignments (i.e., the matching) and corresponding trading prices, which together constitute the OSAAs. In addition to these primary assignments, the auctioneer constructs, for each buyer $b_n$, a preference list $\mathcal{O}_n^t$ that ranks unselected yet feasible sellers according to their suitability. These preference lists serve as backup options during the execution phase, enabling rapid UAV-to-vehicle reassignment without re-running the auction and thereby improving time-efficiency.

\textit{Phase 2. OASS and Preference list Execution.} This phase takes place at the end of timeslot $t$, when participants arrive at the intersection. Its objective is to realize timely and reliable service delivery by executing the pre-signed OSAAs, while handling real-time supply/demand uncertainties. If a pre-established agreement fails or a new demand arises, buyers can immediately switch to alternative sellers using their preference lists, thereby maintaining service continuity without initiating a new auction. This phase consists of the following components:

\noindent$\bullet$\textbf{ OSAA execution:} For each winning buyer–seller pair, the pre-signed OSAA is executed if the buyer's service demand materializes. In this case, the corresponding services are delivered and the agreed payments are completed. If the demand does not materialize, the seller's reserved resources are released and become available for other requests.

\noindent$\bullet$\textbf{ Contingency plan:} To address demand uncertainty and mobility-related disruptions, a lightweight contingency mechanism is activated: if a contracted seller becomes unavailable or a buyer experiences an unforeseen service need, the buyer consults its preference list $\mathcal{O}_n^t$ and selects the highest-ranked available seller. This reassignment is performed without auction re-running, enabling rapid adaptation and ensuring service continuity with minimal delay.

\section{Formal Modeling of LOSA}
This section provides a detailed mathematical formalization of core components. Without loss of generality, we focus on the trading process within a representative timeslot $t$, leveraging the periodic structure of the system.

\subsection{Privacy Objectives and Threat Model}
\label{subsec:threat_model}

Protecting the trajectory privacy of buyers (i.e., vehicles) is a central requirement in AGIN-based service markets. We first formalize the adversarial setting and then outline the design objectives of our privacy-preserving mechanism.

\subsubsection{Adversary Model}
We consider the auctioneer and sellers (UAVs) as \textit{honest-but-curious} (HBC) adversaries as follows: 

\noindent$\bullet$ \textit{Honest:} They faithfully follow the market rules (e.g., Algorithms \ref{alg4_3}-\ref{alg4_4} in Section \ref{subs5.2}) and do not deviate from the specified execution procedures.

\noindent$\bullet$ \textit{Curious:} They may collect and analyze historical transaction data, particularly the reported virtual trajectories $\hat{\bm{\mathcal{L}}}_{n}^{\mathsf{b},t}$, with the goal of inferring the buyers' true trajectories $\bm{\mathcal{L}}_{n}^{\mathsf{b},t}$.

We assume that the adversary possesses strong inference capabilities, characterized by:
\textit{(i) Background knowledge:} The adversary has full knowledge of the road network topology and statistical distribution of traffic flows.
\textit{(ii) Algorithm transparency:} Following Kerckhoffs's principle \cite{6231360}, the location obfuscation mechanism (see Section \ref{sec_location_obfuscation}) is publicly known.
\textit{(iii) Inference capability:} The adversary can employ inference attacks (e.g., Bayesian estimation) to compute the posterior probability that a buyer is located at position $l$, given an observed obfuscated location $\hat{l}$, denoted as $\Pr(l | \hat{l})$.

\subsubsection{Threat Method}
We focus on \textit{location inference} as the primary privacy threat. Due to the road network constraints, vehicle mobility is highly structured. In particular, a vehicle cannot occupy arbitrary spatial locations, as its movement is restricted to valid road segments and intersections. As a result, naive location perturbation may generate physically infeasible positions, which can be easily identified and discarded by an adversary, thereby weakening the effectiveness of obfuscation. Moreover, the adversary can exploit temporal correlations in the reported trajectory data: by observing a sequence of obfuscated locations over time $\mathcal{T}$, the adversary can employ sequential inference techniques, such as Markov models \cite{10978388}, to refine estimates of the true trajectory. This temporal aggregation reduces uncertainty and enables the adversary to recover sensitive mobility patterns, including frequently visited locations and travel routines.

\subsubsection{Privacy Design Goals}
To defend against the aforementioned adversaries, traditional heuristic privacy methods (e.g., $k$-anonymity or simple additive noise \cite{10296013}) are insufficient as they often do not provide rigorous guarantees upon having adversaries with background knowledge and inference capabilities.
We adopt the principle of \textit{DP}, specifically \textit{Geo-I}, which is well suited for spatial data.
Our goal is to ensure that geographically nearby locations remain statistically indistinguishable after obfuscation. In particular, for any two locations $l$ and $l'$ in the true trajectory, the probability of generating the same obfuscated location $\hat{l}$ should satisfy:
\begin{equation}
	\label{eq:dp_goal}
	\Pr(\hat{l} | l) \le e^{\epsilon \cdot d(l, l')} \Pr(\hat{l} | l'),
\end{equation}
where $d(\cdot)$ is a distance metric and $\epsilon$ corresponds to a privacy budget ($\epsilon = \xi_{n}^{t}$). This guarantee ensures that, regardless of the adversary's prior knowledge, the ability to distinguish whether a user is located at $l$ or $l'$ is bounded by $\epsilon$. As a result, smaller values of $\epsilon$ provide stronger privacy protection, while larger values allow for higher utility.

We design a \textit{dynamic privacy budget} $\xi_t^n$, together with a \textit{polar-coordinate-based obfuscation scheme} (Section \ref{sec_location_obfuscation}). This design enables the adaptive control of privacy levels over time, allowing each user to balance privacy protection and service utility in response to the observed system outcomes (i.e., matchings).

\subsection{Modeling of OSAAs and Preference Lists}
We proceed to model two key components of the proposed framework, namely OSAAs and preference lists.

\textbf{\textit{Modeling of OSAAs:}} An OSAA $\mathbb{C}_{m,n}^{t}$ represents a binding agreement established during Phase 1 of LOSA between a buyer $b_n$ and a seller $s_m$ for service delivery at the end of timeslot $t$; mathematically, it is defined as:
\begin{equation}
	\mathbb{C}_{m,n}^{t}=\left\{ \bm{p}_{m,n}^{t},\bm{r}_{m,n}^{t},\mathbb{L}_{m,n}^{t} \right\},
\end{equation}
which consists of the following components:

\noindent\textit{(i) The payment vector}  $\bm{p}_{m,n}^{t}=\left\{p_{m,n}^{1,t},\ldots,p_{m,n}^{j,t},\ldots,p_{m,n}^{J,t}\right\}$ specifies the buyer's payment for each potential service type, with $p_{m,n}^{j,t}$ reflecting the unit price that buyer $b_n$ agreed to pay for service type $j$ offered by $s_m$. These prices can depend on factors, such as service quality and market conditions.

\noindent\textit{(ii) The revenue vector}  $\bm{r}_{m,n}^{t}=\left\{r_{m,n}^{1,t},\ldots,r_{m,n}^{j,t},\ldots,r_{m,n}^{J,t}\right\}$ defines the corresponding revenue for the seller. Each element $r_{m,n}^{j,t}$ reflects the compensation received for providing service type $j$.

\noindent\textit{(iii) The spatio-temporal anchor}  $\mathbb{L}_{m,n}^{t}$ specifies the intersection at which the service is to be delivered, thereby enforcing spatial consistency with the mobility model.
\begin{figure*}[b]
	\centering
	\hrulefill
	\begin{equation}\label{eq3}
		\begin{aligned}
			\xi_t^{n+1} = \xi_t^n - \underbrace{\eta  \tanh\left(\gamma  \Delta U_{n}^{t}\right) \times\left(1 - \frac{\xi_t^n}{\xi_{\max}}\right)}_{\text{Utility gradient term}} + \underbrace{\theta  C_n^t  \left(\xi_{\max} - \xi_t^n\right)}_{\text{Matching failure penalty term}} + \underbrace{\epsilon^{t}}_{\text{Stochastic noise term}},
		\end{aligned}
	\end{equation}
	\vspace{-5mm}
\end{figure*}
\textbf{\textit{Modeling of preference lists:}} To enhance robustness against demand and mobility uncertainty, each buyer $b_n$ maintains a preference list $\bm{\mathcal{O}}_{n}^{t} = \{\mathcal{O}_{n}^{1,t}, \dots,\mathcal{O}_{n}^{j,t}, \dots, \mathcal{O}_{n}^{J,t}\}$ as a contingency mechanism. For each service type $j$, sub-list $\mathcal{O}_{n}^{j,t}$ contains an ordered set of candidate sellers that are not selected in OSAA but remain feasible alternatives. The construction of each sub-list follows two criteria: \textit{(i)} all sellers on the list should meet a minimum threshold on service quality for type $j$, and \textit{(ii)} they are ranked in non-decreasing order of their asking prices. This structure creates a price-ordered candidate pool that enables efficient reassignment during execution (i.e., Phase 2 of LOSA). If an OSAA cannot be fulfilled, e.g., due to seller unavailability, the buyer can immediately select the highest-ranked available seller from its preference list without re-running the auction.

\subsection{Modeling of Buyers}
\label{subsec:buyer_modeling}
We model each buyer $b_n \in \bm{\mathcal{B}}$ as a strategic agent whose behavior is driven by mobility, privacy preferences, economic considerations, and uncertain service demands. For any timeslot $t$, the state and characteristics of buyer $b_n$ are captured by the following tuple:
\begin{equation}
	\label{eq2}
	b_n = \left( 
	\bm{\mathcal{L}}_{n}^{\mathsf{b},t},\ 
	\hat{\bm{\mathcal{L}}}_{n}^{\mathsf{b},t},\ 
	\mathbbm{r}_n,\  
	\xi_{n}^{t},\ 
	\bm{\mathcal{V}}_n,\ 
	\bm{\mathcal{K}}_n,\ 
	\bm{bid}_{n}^{t},\ 
	\bm{\mathcal{Q}}_{n}^{t} 
	\right),
\end{equation}
which captures the key factors influencing its decision-making. The components of (\ref{eq2}) are described in the following subsections.

\subsubsection{Privacy-aware Mobility and Dynamic Budget}
The mobility and privacy behavior of buyers, represented by vehicles, plays a central role in the proposed framework and are detailed below.

\noindent $\bullet$ \textit{Trajectories and Privacy Radius:} Each buyer follows a \textit{true trajectory} $\bm{\mathcal{L}}_{n}^{\mathsf{b},t} = \left\{l_{n}^{\mathsf{b},t}, l_{n}^{\mathsf{b},t+1}, \ldots, l_{n}^{\mathsf{b},|\bm{\mathcal{T}}|}\right\}$ starting from timeslot $t$, which is treated as private information. To join the resource trading market while protecting its privacy, the buyer discloses a \textit{virtual trajectory} $\hat{\bm{\mathcal{L}}}_{n}^{\mathsf{b},t} = \left\{l_{n}^{\mathsf{b},t}, \hat{l}_{n}^{\mathsf{b},t+1}, \ldots, \hat{l}_{n}^{\mathsf{b},|\bm{\mathcal{T}}|}\right\}$. The current location $l_{n}^{\mathsf{b},t}$ is reported truthfully to maintain contractual feasibility, whereas future locations are deliberately perturbed. This perturbation is confined within a predefined \textit{privacy radius} $\mathbbm{r}_n$, which restricts the deviation between true and virtual trajectories. This radius ensures that the reported trajectory remains physically plausible while providing a controllable level of privacy protection.

\noindent $\bullet$ \textit{Dynamic Privacy Budget:} The extent of location obfuscation is regulated by the dynamic privacy budget $\xi_{n}^{t}$, which controls the trade-off between privacy and utility for buyer $b_n$ at timeslot $t$. $\xi_{n}^{t}$ determines the level of indistinguishability in the Geo-I mechanism and reflects the buyer's tolerance for location disclosure. A larger $\xi_{n}^{t}$ results in weaker perturbation, producing a virtual trajectory that more closely aligns with the true trajectory, thereby improving matching accuracy and service quality. Conversely, a smaller $\xi_{n}^{t}$ enforces stronger obfuscation, enhancing privacy protection while potentially degrading service effectiveness due to reduced spatial precision.

To capture the adaptive nature of privacy preference of each buyer, we model the evolution of the privacy budget as a \textit{feedback-driven stochastic process}. Specifically, the budget for the next timeslot $\xi_t^{n+1}$ is updated according to the outcomes of the current timeslot $t$ as (\ref{eq3}), where $\Delta U_{n}^{t} \triangleq (U_{n}^{\mathsf{b},t} - \overline{U}_{n}^{\mathsf{b},t-1}) / \overline{U}_{n}^{\mathsf{b},t-1}$ is the \textit{normalized utility variation}, with $\overline{U}_{n}^{\mathsf{b},t-1}$ being the average utility over the past $K$ slots.
Moreover, $C_n^t$ represents the \textit{weighted degree of matching failures} within a sliding window of length $K$ defined as:
\begin{equation}
	C_n^t \triangleq \sum_{\tau=t-K}^{t} \left(1 - \frac{1}{J}\sum_{j=1}^J \mathbb{I}(\exists m: \ x_{m,n}^{j,\tau}=1)\right),
\end{equation}
where $x_{m,n}^{j,\tau} \in \{0,1\}$ denotes the binary matching variable indicating whether buyer $b_n$ is matched with seller $s_m$ for service type $j$ in timeslot $\tau$, and $\mathbb{I}(\cdot)$ denotes an indicator function that returns 1 if a matching is found for service type $j$ in timeslot $\tau$. This formulation captures demand-level service failures in a proportional manner across different service types. In addition, we incorporate a stochastic term $\epsilon^{t} \sim \mathcal{N}(0,\sigma^2)$ to model unobserved environmental factors and behavioral variability. The introduced update mechanism in (\ref{eq3}) is driven by the following principles:

\noindent\textit{(i) Utility inverse regulation:} The privacy budget evolves inversely with recent utility trends. When utility improves, the buyer tends to reduce $\xi_{n}^{t}$ to enhance privacy protection. Conversely, when utility deteriorates, the buyer increases $\xi_{n}^{t}$ to improve matching accuracy in future interactions.

\noindent\textit{(ii) Failure accumulation penalty:} A higher value of $C_n^t$, indicating frequent matching failures, leads to an increase in the privacy budget. This encourages less aggressive obfuscation to improve service reliability.

\noindent\textit{(iii) Stochastic evolution:} The noise term prevents the update process from converging prematurely to suboptimal operating points, allowing continued exploration of the privacy–utility trade-off.

\subsubsection{Economic Attributes and Stochastic Demand}
Economic attribute of a buyer is characterized by its valuations, costs, bids, and the stochastic nature of its service demand. These components are described below.

\noindent $\bullet$ \textit{Valuation, privacy cost, and bids:}
The \textit{valuation vector} $\bm{\mathcal{V}}_n = \left\{v_{n}^{1},\ldots,v_{n}^{J}\right\}$ represents the intrinsic value that the buyer assigns to each service type under ideal conditions. The \textit{privacy cost vector} $\bm{\mathcal{K}}_n = \left\{k_{n}^{1},\ldots,k_{n}^{J}\right\}$ quantifies the sensitivity of the buyer to privacy loss, where $k_{n}^{j}$ denotes the unit cost associated with revealing location information for service type $j$. The \textit{bid vector} $\bm{bid}_{n}^{t} = \left\{bid_{n}^{1,t},\ldots,bid_{n}^{J,t}\right\}$ specifies the willingness-to-pay reported by the buyer to the auctioneer for each service type at timeslot $t$.

\noindent $\bullet$ \textit{Stochastic demand modeling:} In practical AGIN environments, a buyer's service demand is inherently dynamic and uncertain. To capture this behavior for each buyer $b_n$, we model its demand for service type $j$ at timeslot $t$ as a Bernoulli random variable: $Q_{n}^{j,t} \sim \mathrm{Ber}\left(\mathbbm{q}_{n}^{j,t}\right)$, where $\mathbbm{q}_{n}^{j,t}$ denotes the probability that buyer $b_n$ requests service type $j$ at timeslot $t$. We model this demand probability over time based on past service outcomes according to:
\begin{equation}
	\hspace{-4mm}\mathbbm{q}_{n}^{j,t+1} = 
	\begin{cases}
		\mathbbm{q}_{n}^{j,t} \cdot e^{-\lambda}, & 
		\begin{aligned}
			&\text{if the service is fulfilled} \\
			&\text{at timeslot } t;
		\end{aligned} \\
		\mathbbm{q}_{n}^{j,t} + \beta (1-\mathbbm{q}_{n}^{j,t}), & 
		\text{otherwise}.
	\end{cases}
	\hspace{-4mm}
\end{equation}
Here, $\lambda$ is a decay rate and $\beta$ is a reinforcement coefficient. This update captures a \textit{memory-dependent demand evolution}: when a service request is fulfilled, the demand probability decreases due to temporary satisfaction. In contrast, when demand is unmet, the probability increases, reflecting growing urgency. This mechanism enables the model to capture temporal dependencies in user behavior. Accordingly, we define the comprehensive demand state vector for buyer $b_n$ across all $J$ service types at timeslot $t$ as $\bm{\mathcal{Q}}_{n}^{t} = \left\{\mathbbm{q}_{n}^{1,t}, \dots, \mathbbm{q}_{n}^{J,t}\right\}$.

Among the above attributes, the true trajectory $\bm{\mathcal{L}}_{n}^{\mathsf{b},t}$, valuation $\bm{\mathcal{V}}_n$, and privacy cost sensitivity $\bm{\mathcal{K}}_n$ are private information, while the virtual trajectory $\hat{\bm{\mathcal{L}}}_{n}^{\mathsf{b},t}$, bid $\bm{bid}_{n}^{t}$, and the aforementioned demand model parameters $\bm{\mathcal{Q}}_{n}^{t}$ are disclosed to the auctioneer.

\subsubsection{Buyer Utility Formulation}
We denote the utility of buyer $b_n$ during timeslot $t$ by $U_{n}^{\mathsf{b},t}$, which integrates the above discussed attributes to quantify the net benefit obtained from service provisioning. For a given matching decision $x_{m,n}^{j,t}$ (where $x_{m,n}^{j,t}=1$ if $b_n$ is matched with $s_m$ for service $j$, and 0 otherwise), $U_{n}^{\mathsf{b},t}$ is given by
\begin{equation}
	\label{eq7}
	U_{n}^{\mathsf{b},t}=\sum_{j\in J}\sum_{s_m\in \bm{\mathcal{S}}} x_{m,n}^{j,t}\left[Q_{n}^{j,t}\left(\Gamma_{m,n}^{t}v_{n}^{j}-k_{n}^{j} \xi_{n}^{t}-p_{m,n}^{j,t}\right)\right].
\end{equation}
This formulation consists of the following components:

\noindent\textit{(i) Service benefit ($\Gamma_{m,n}^{t}v_{n}^{j}$):} The intrinsic valuation $v_{n}^{j}$ is scaled by \textit{trajectory similarity} $\Gamma_{m,n}^{t}$, which reflects the spatial alignment between the buyer's true trajectory and the seller's historical trajectory (i.e., a higher similarity indicates better service quality and thus higher realized value.) In particular, $\Gamma_{m,n}^{t}$ is computed via the normalized Fréchet distance as follows:
\begin{equation}
	\label{eq8}
	\Gamma_{m,n}^{t} = \left(1 - \frac{F(\bm{\mathcal{L}}_{n}^{\mathsf{b},t},\bm{\mathcal{L}}_{m}^{\mathsf{s},t})}{\max\left(\text{Leg}_{n}^{\mathsf{b},t}, \text{Leg}_{m}^{\mathsf{s},t}\right)}\right) \times 100\%,
\end{equation}
where $F(\cdot,\cdot)$ is the Fréchet distance and $\text{Leg}$ denotes the total trajectory length.

\noindent\textit{(ii) Privacy cost ($k_{n}^{j} \xi_{n}^{t}$):} This term quantifies the disutility from privacy loss, modeled as a linear function of the dynamic privacy budget $\xi_{n}^{t}$ and the buyer's sensitivity $k_{n}^{j}$.

\noindent\textit{(iii) Payment ($p_{m,n}^{j,t}$):} This represents the price paid by buyer $b_n$ for receiving service type $j$ from seller $s_m$.

Note that in (\ref{eq7}), the utility is modulated by the stochastic demand $Q_{n}^{j,t}$, meaning that the realized utility in (\ref{eq7}) is only obtained when the corresponding service is actually required. However, decisions in Phase 1 of LOSA are made prior to the realization of $Q_{n}^{j,t}$. Therefore, buyers optimize their \textit{expected utility}, as given by
\begin{equation}
	\overline{U_{n}^{\mathsf{b},t}} = \sum_{j\in J}\sum_{s_m\in \bm{\mathcal{S}}} x_{m,n}^{j,t}\left[\mathbbm{q}_{n}^{j,t}\left(\Gamma_{m,n}^{t}v_{n}^{j}-k_{n}^{j} \xi_{n}^{t}-p_{m,n}^{j,t}\right)\right].
\end{equation}
This expected utility serves as the objective function guiding the buyer's decision-making and bidding behavior in the proposed auction mechanism.

\subsection{Modeling of Sellers}
In our network of interest, UAVs serve as sellers, offering different types of services to vehicles. For each seller $s_m \in \bm{\mathcal{S}}$, its state and characteristics for a given timeslot $t$ are captured by the following triple:
\begin{equation}
	s_m=\left(\bm{\mathcal{L}}_{m}^{\mathsf{s},t}, \bm{\mathcal{C}}_m, \bm{ask}_m^{t}\right),
\end{equation}
capturing its mobility, cost structure, and pricing behavior. The components of this tuple are described below.

\subsubsection{Proactive Trajectory}
The seller’s mobility is represented by its proactive trajectory $\bm{\mathcal{L}}_{m}^{\mathsf{s},t}$, recording the locations visited up to the current timeslot $t$. Rather than serving as a static historical trace, this trajectory is continuously updated to reflect the seller's evolving movement pattern. Specifically, at each timeslot, the seller appends its current location to $\bm{\mathcal{L}}_{m}^{\mathsf{s},t}$, enabling an incremental representation of its mobility. This trajectory serves two key purposes: \textit{(i)} it provides buyers with sufficient path history to evaluate trajectory similarity, which is a critical factor in determining service quality and matching decisions; and \textit{(ii)} it enables sellers to adapt their future movements based on observed demand patterns. By analyzing historical interactions, a seller can proactively steer its trajectory toward regions with higher vehicle density, thereby increasing the likelihood of successful transactions.

\subsubsection{Economic Attributes}
The economic profile of each seller is characterized by its service costs and declared asking prices, as detailed below.

\noindent$\bullet$ \textit{The cost vector} $\bm{\mathcal{C}}_m=\left\{c_{m}^{1},\ldots,c_{m}^{J}\right\}$ quantifies the expenditure associated with providing each service type, including factors such as flight energy, communication overhead, and computational load. This information is private to the seller and influences its pricing behavior.

\noindent$\bullet$ \textit{The asking price vector} $\bm{ask}_m^{t}=\left\{{ask}_{m}^{1,t},\ldots,{ask}_{m}^{J,t}\right\}$ represents the prices declared by the seller for each service type at timeslot $t$. These values are submitted to the auctioneer during Phase 1 of LOSA and determine the seller's participation in the matching process.

\subsubsection{Seller Utility Formulation}
The utility of seller $s_m$ at timeslot $t$, denoted by $U_{m}^{\mathsf{s},t}$, is defined as the net profit obtained from service provisioning, i.e., the difference between revenue and operational cost. For a given matching decision $x_{m,n}^{j,t}$, this utility is expressed as
\begin{equation}
	\label{eq10}
	U_{m}^{\mathsf{s},t}=\sum_{j\in J}\sum_{b_n\in \bm{\mathcal{B}}} x_{m,n}^{j,t}\left[Q_{n}^{j,t}\left(r_{m,n}^{j,t}-c_{m}^{j}\right)\right],
\end{equation}
where $r_{m,n}^{j,t}$ denotes the revenue received for offering service type $j$.
Similar to that of buyers, since our designed auction in Phase 1 of LOSA occurs before the buyer's demand $Q_{n}^{j,t}$ is realized, sellers optimize their \textit{expected utility}, as given by
\begin{equation}
	\overline{U_{m}^{\mathsf{s},t}}=\sum_{j\in J}\sum_{b_n\in \bm{\mathcal{B}}} x_{m,n}^{j,t}\left[\mathbbm{q}_{n}^{j,t}\left(r_{m,n}^{j,t}-c_{m}^{j}\right)\right].
\end{equation}
This formulation highlights a dual objective for sellers. On the one hand, the sellers aim to maximize immediate profit through favorable transaction margins $(r_{m,n}^{j,t} - c_{m}^{j})$. On the other hand, they can strategically influence future opportunities by adjusting their trajectory $\bm{\mathcal{L}}_{m}^{\mathsf{s},t}$ toward regions with higher expected demand, thereby increasing the likelihood of future matches through elevated $\mathbbm{q}_{n}^{j,t}$.

We further observe an important structural property of the proposed framework: the seller's proactive trajectory planning complements the buyer's privacy-preserving behavior in an asymmetric manner. Specifically, buyers protect their mobility patterns by reporting obfuscated trajectories $\hat{\bm{\mathcal{L}}}_{n}^{\mathsf{b},t}$, whereas sellers reveal their actual trajectories $\bm{\mathcal{L}}_{m}^{\mathsf{s},t}$ to facilitate matching. This asymmetric information structure enables efficient spatio-temporal coordination while preserving vehicles' privacy.

\subsection{Modeling of Auctioneer}
The trading process is coordinated by an auctioneer, implemented as a trusted platform (e.g., edge/cloud infrastructure) and instantiated through RSUs deployed at intersections. At each timeslot, RSUs execute the two-phase trading mechanism of LOSA with their roles and responsibilites detailed below.

\subsubsection{Role and Responsibility}
The responsibilities of the auctioneer during each phase of LOSA are as follows:

\noindent $\bullet$ \textit{Phase 1 (planning and agreement formation):} The auctioneer collects and validates bid vectors $\bm{bid}_{n}^{t}$ from buyers and ask vectors $\bm{ask}_m^{t}$ from sellers. Based on this information, it determines proper buyer-seller matchings ($\bm{X}^{t}$, denoting the global matching configuration), constructs the corresponding OSAAs ($\mathbb{C}_{m,n}^{t}$), and generates preference lists ($\bm{\mathcal{O}}_{n}^{t}$).

\noindent $\bullet$ \textit{Phase 2 (execution and coordination):} The auctioneer oversees the execution of pre-established OSAAs and monitors their fulfillment. When disruptions occur, such as unmet demand or seller unavailability, it facilitates supplementary transactions using the predefined preference lists, ensuring that all interactions comply with the market rules.

\subsubsection{Auctioneer Utility Formulation}
\label{sec:AuctioneerUtility}
The overall economic benefit of auctioneer, denoted by $U^{\mathsf{a},t}$, reflects the value that it extracts from coordinating the market. It is defined as the net margin between payments collected from buyers and revenues disbursed to sellers:
\begin{equation}
	\label{eq4_6}
	U^{\mathsf{a},t}=\sum_{j \in J}\sum_{s_m\in \bm{\mathcal{S}}}\sum_{b_n\in \bm{\mathcal{B}}} x_{m,n}^{j,t}\left[Q_{n}^{j,t}\left(p_{m,n}^{j,t}-r_{m,n}^{j,t}\right)\right],
\end{equation}
where $p_{m,n}^{j,t}$ is the payment collected from buyer $b_n$, and $r_{m,n}^{j,t}$ is the revenue paid to seller $s_m$. The term $(p_{m,n}^{j,t}-r_{m,n}^{j,t})$ represents the auctioneer's margin over transactions, and the utility is aggregated over all valid trades, conditioned on the buyer's demand $Q_{n}^{j,t}$ materialization. Since the demand $Q_{n}^{j,t}$ is stochastic and unknown at the decision phase, the auctioneer instead optimizes its \textit{expected utility}, as given by

\begin{equation}
	\label{eq4_7}
	\overline{U^{\mathsf{a},t}}=\sum_{j \in J}\sum_{s_m\in \bm{\mathcal{S}}}\sum_{b_n\in \bm{\mathcal{B}}} x_{m,n}^{j,t}\left[\mathbbm{q}_{n}^{j,t}\left(p_{m,n}^{j,t}-r_{m,n}^{j,t}\right)\right].
\end{equation}
This formulation highlights the auctioneer's role as a market coordinator that balances revenue extraction with efficient matching: by selecting the matching configuration $\bm{X}^{t}$ to maximize expected utility, the auctioneer implicitly promotes high-probability transactions while avoiding inefficient or unstable matchings.

\subsection{Location Obfuscation Model}
\label{sec_location_obfuscation}
To operationalize privacy preservation for each buyer, we design a mechanism for generating the virtual trajectory $\hat{\bm{\mathcal{L}}}_{n}^{\mathsf{b},t}$ that achieves a controllable balance between location privacy and service effectiveness. Specifically, we design a discrete polar-coordinate-based obfuscation mechanism that generates the virtual trajectory $\hat{\bm{\mathcal{L}}}_{n}^{\mathsf{b},t}$ within a maximum allowable privacy radius $\mathbbm{r}_n$. Specifically, the mechanism sequentially samples a perturbation radius based on a logarithmic probability distribution tightly tuned by the dynamic budget $\xi_{n}^{t}$, and subsequently determines a discretized angular component to finalize the location. The detailed mathematical formulations and step-by-step procedures of this obfuscation mechanism are deferred to Appendix \ref{appendix:obfuscation}. Although the final service execution occurs at discrete intersections, this fine-grained location obfuscation fundamentally impacts both the auction pricing and the defense against continuous trajectory tracking, as discussed in Appendix \ref{appendix:conceptual_discussion}.

\subsection{Privacy Guarantee Analysis}
\label{sec:privacy_analysis}
To quantify the privacy protection achieved by our LOSA, we establish a formal bound on location leakage in terms of the dynamic privacy budget $\xi_{n}^{t}$ as follows:

\begin{thm}[$\epsilon$-Geo-I]
	\label{thm:geo_indistinguishability}
	The randomized obfuscation mechanism $\mathcal{A}$ defined by the discretized polar-coordinate sampling procedure (Eqs. (\ref{eq4_8})-(\ref{eq4_10}) in Appendix \ref{appendix:obfuscation}) satisfies $\epsilon$-Geo-I, where the privacy level $\epsilon$ ($\epsilon = \xi_{n}^{t}$) is strictly bounded by the dynamic privacy budget $\xi_{n}^{t}$. Specifically, for any two true locations $l, l'$ and any obfuscated location $\hat{l}$, we have:
	\begin{equation}
		\label{eq:geo_indistinguishability}
		\frac{\Pr(\mathcal{A}(l) = \hat{l})}{\Pr(\mathcal{A}(l') = \hat{l})} \leq e^{\xi_{n}^{t} \cdot d(l, l')}.
	\end{equation}
\end{thm}
\begin{proof}[Proof Sketch]
	The proof analyzes how the sampling distribution changes with the true location. Since the radius distribution in (\ref{eq4_9}) is logarithmic in the distance to the true location, its log-probability can be shown to satisfy a Lipschitz condition under the distance metric. Next, applying the Mean Value Theorem to the logarithmic term allows us to bound the changes in probability when $l$ shifts to $l'$. The triangle inequality on $d(\cdot,\cdot)$ is then used to relate this change to the distance $d(l,l')$. Combining these results yields an upper bound on the likelihood ratio of generating the same obfuscated location from two different true locations, controlled by $\xi_{n}^{t}$ (see the full proof in Appendix \ref{appendix:proof_theorem1}).
\end{proof}

\subsection{Market Design Properties}
LOSA aims to satisfy several economic properties that ensure reliable, efficient, and fair operation in dynamic AGIN environments: \textit{Individual Rationality (IR)}, \textit{Budget Balance (BB)}, \textit{Truthfulness (Incentive Compatibility, IC)}, and \textit{Computational Efficiency (CE, reflecting the timeliness)}, with detailed definitions in Appendix \ref{Market Design Properties}.

\section{LOSA: A Proactive and Privacy-Preserving Double Auction Framework over AGINs}

This section formalizes LOSA (also see Appendix \ref{alg4}). We first develop the look-ahead and privacy-aware double auction in Phase 1 to generate OSAAs and buyer preference lists, and then describe how these decisions are executed in Phase 2. As Phase 2 only executes pre-established and fallback decisions, the technical focus is mostly on Phase 1.

\subsection{Problem Modeling for Phase 1}
Social welfare (SW) is a fundamental metric for assessing overall efficiency of a service trading market. Accordingly, we define SW as the aggregate utility as follows:
\begin{equation}
	\label{eq20}
	U^{\mathsf{sw}}=\sum_{t\in \bm{\mathcal{T}}}\left(\sum_{b_n\in \bm{\mathcal{B}}}U_{n}^{\mathsf{b},t}+\sum_{s_m\in \bm{\mathcal{S}}}U_{m}^{\mathsf{s},t}+U^{\mathsf{a},t}\right).
\end{equation}
Reaching the global SW optimum is computationally intractable due to the strong coupling induced by long-term system dynamics. Nevertheless, we face sequential state dependency: the feasible decision space in each timeslot (e.g., network topology and reachable participants) is determined by the system state at the end of the previous timeslot. This enables us to reduce the global SW maximization to a series of tractable per-timeslot optimization subproblems. Accordingly, for a representative timeslot $t$, we reformulate:
\begin{equation}
	\label{eq_4_12}
	\hspace{-3mm}U^{\mathsf{sw},t}{=}\sum_{j\in J}\sum_{b_n\in \bm{\mathcal{B}}} \sum_{s_m\in \bm{\mathcal{S}}} x_{m,n}^{j,t} Q_{n}^{j,t}\left(\Gamma_{m,n}^{t}v_{n}^{j}-k_{n}^{j} \xi_{n}^{t}-c_{m}^{j}\right)\hspace{-.3mm},\hspace{-3mm}
\end{equation}
where $\Gamma_{m,n}^{t}v_{n}^{j}$ represents the value gain from service matching, $k_{n}^{j} \xi_{n}^{t}$ denotes the privacy cost, and $c_{m}^{j}$ indicates service provisioning cost. This formulation reveals three optimization dimensions: improving matching quality, reducing privacy costs, and enhancing service delivery efficiency. Moreover, since Phase 1 is conducted prior to demand realization, we define the expected social welfare (ESW) as:
\begin{equation}
	\label{eq_4-13}
	\hspace{-3mm}\overline{U^{\mathsf{sw},t}}{=}\sum_{j\in J}\sum_{b_n\in \bm{\mathcal{B}}}\sum_{s_m\in \bm{\mathcal{S}}}x_{m,n}^{j,t}\mathbbm{q}_{n}^{j,t}\left(\Gamma_{m,n}^{t}v_{n}^{j}-k_{n}^{j} \xi_{n}^{t}-c_{m}^{j}\right),\hspace{-3mm}
\end{equation}
which converts stochastic decision-making into a deterministic optimization through the demand probabilities $\mathbbm{q}_{n}^{j,t}$. Accordingly, $\bm{\mathcal{P}}_1$ for Phase 1 is given by
\begin{equation}
	\label{eq4_14}
	\bm{\mathcal{P}}_1:\mathop{\text{argmax}}\limits_{\bm{X}^{t},\mathbb{C}_{m,n}^{t},\bm{\mathcal{O}}_{n}^{t}} \overline{U^{\mathsf{sw},t}}\\
\end{equation}
\vspace{-5mm}
\begin{subequations}\label{p1}
	\begin{align}			
		\text{s.t.}~~&
		0\leq \sum_{s_m\in\bm{\mathcal{S}}}x_{m,n}^{j,t}\le 1,\ \forall\ b_n\in\bm{\mathcal{B}}, j\in J,\tag{C1}\\
		&0\leq \sum_{b_n\in\bm{\mathcal{B}}}x_{m,n}^{j,t}\le 1,\ \forall\ s_m\in\bm{\mathcal{S}}, j\in J,\tag{C2}\\
		&x_{m,n}^{j,t}\in\left\{0,1\right\},\ \forall\ b_n\in\bm{\mathcal{B}},\forall\ s_m\in\bm{\mathcal{S}}.\tag{C3}\\
		&\hspace{-7mm}\Gamma_{m,n}^{t}v_{n}^{j}-k_{n}^{j} \xi_{n}^{t} \ge c_{m}^{j},\ \forall\ b_n\in\bm{\mathcal{B}},\forall\ s_m\in\bm{\mathcal{S}},x_{m,n}^{j,t}=1.\tag{C4}
	\end{align}
\end{subequations}

\noindent In $\bm{\mathcal{P}}_1$, constraints (C1)-(C2) enforce one-to-one matching per service type, while (C3) ensures binary matching decisions. Constraint (C4) guarantees that each selected match contributes positively to the ESW by requiring the buyer's net valuation to exceed the seller's cost. Problem $\bm{\mathcal{P}}_1$ is a \textit{mixed-integer non-linear programming (MINLP)} problem (where matching variables $x_{m,n}^{j,t}$ are discrete binary, while prices and ESW are continuous), which is typically NP-hard \cite{10681251}. Its complexity stems from two main aspects: \textit{(i)} the binary nature of elements in $\bm{X}^{t}$ leads to a combinatorial search space; and \textit{(ii)} the nonlinear coupling between privacy preservation and matching accuracy induces pronounced non-linear coupling effects. To address these challenges, we develop three core modules in Phase 1 of LOSA as described next.

\subsection{Solution Design for Phase 1}
\label{subs5.2}  
We propose a decision engine consisting of three synergistic modules (see more details in Appendix \ref{alg4}): Module A performs demand-aware UAV repositioning via mobility prediction. Module B determines ESW-maximizing buyer–seller matching under updated system states. Module C implements VCG-based pricing to ensure IC, IR, BB, and CE. Together, they jointly optimize positioning, allocation, and payment under privacy–efficiency tradeoffs.

\begin{figure}[]
	\centering
	\includegraphics[trim=0cm 0cm 0cm 0cm, clip, width=\columnwidth]{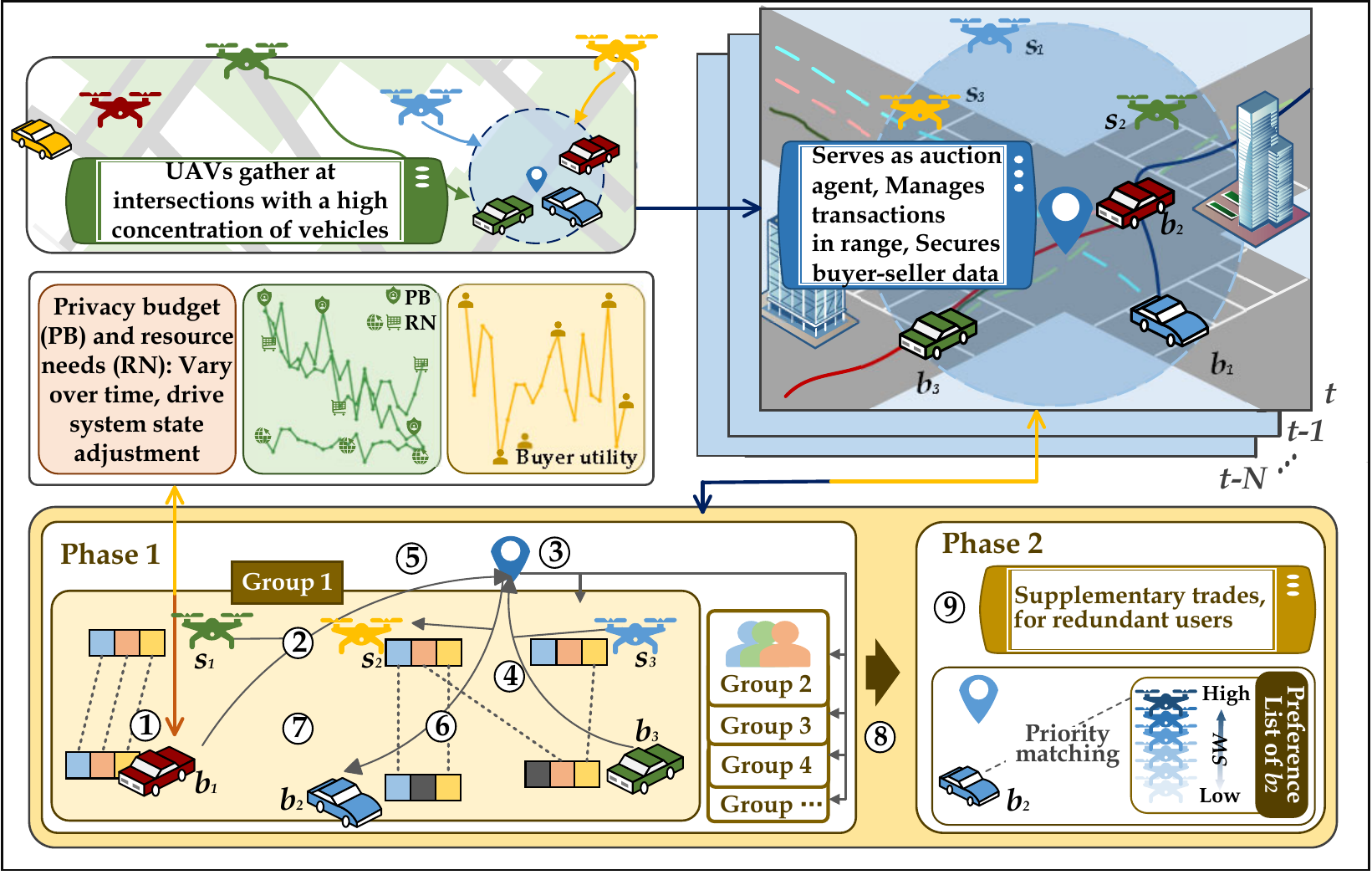} 
	\vspace{-7mm}
	\caption{Diagram of the overall workflow of the LOSA mechanism, integrating modules A, B, and C. 
		Following the UAV deployment in Module A, the workflow proceeds as: \textit{(i)} Module B (Buyer--Seller Matching): covers steps \ding{172}--\ding{176}, where \ding{172} buyers update dynamic parameters (privacy budget $\xi_{n}^{t}$ and demand $\bm{\mathcal{Q}}_{n}^{t}$); \ding{173}--\ding{174} participants upload trajectories (e.g., virtual path $\hat{\bm{\mathcal{L}}}_{n}^{\mathsf{b},t}$) for similarity-based grouping; and \ding{175}--\ding{176} the auctioneer collects bids/asks to execute the designed auction process. \textit{(ii)} Module C (VCG-based Pricing): covers steps \ding{177}--\ding{178}, including \ding{177} generating binding OSAAs $\mathbb{C}_{m,n}^{t}$ with defined service matching and pricing structures; and \ding{178} generating preference lists $\bm{\mathcal{O}}_{n}^{t}$ for backup. \textit{(iii)} Phase 2 (Real-time Execution): covers steps \ding{179}--\ding{180}, where \ding{179} participants proceed to the designated intersection; and \ding{180} RSUs orchestrate supplementary trades for contingency needs (e.g., backup users).}
	\label{fig_3}
\end{figure}

\subsubsection{Module A: UAV Trajectory Planning for the Next Timeslot}

Since the destination of buyers at the end of the next timeslot is determined by their mobility patterns, we proactively guide sellers toward favorable locations to meet anticipated demand. Algorithm \ref{alg4_2} implements a distributed decision-making mechanism that enables each UAV to optimize its trajectory with the objective of maximizing the ESW. The procedure consists of three main steps:

\noindent$\bullet$ \textit{Step 1: Topology identification (lines 3-7).} 
For each seller $s_m \in \bm{\mathcal{S}}$, the algorithm first retrieves its current location at timeslot $t-1$. Based on the grid-structured road network, it then determines the set of feasible intersections that the UAV can reach in the next timeslot, denoted by $N^{locs}$. This candidate set typically includes the current intersection and its four immediate neighbors (i.e., north, south, east, and west), reflecting the mobility constraints imposed by the underlying road topology.

\noindent$\bullet$ \textit{Step 2: Marginal contribution evaluation (lines 8-13).} 
For every candidate intersection $l \in N^{locs}$, the seller evaluates the benefit of relocating to $l$ through its \textbf{marginal contribution} to ESW. To avoid the complexity of global optimization, this evaluation is performed locally. Specifically, the seller first estimates the set of competing peers at location $l$, denoted by $\bm{\mathcal{S}}'_l$ (i.e., other sellers that might also move to $l$). Then, it invokes Algorithm \ref{alg4_3} to calculate ESW at intersection $l$ under two scenarios: \textit{(i)} including $s_m$, and \textit{(ii)} excluding $s_m$. The difference between these two values defines the marginal contribution, denoted by $\text{Mar}^\text{Con}$, which quantifies the net gain in ESW attributable to $s_m$ choosing location $l$.

\noindent$\bullet$ \textit{Step 3: Dual-criteria decision making (lines 14-18).} 
To determine the final target location, the seller employs a two-layer selection logic. It first filters the candidate locations to find those offering the maximum $\text{Mar}^\text{Con}$. If multiple locations achieve the same maximum value, a secondary tie-breaking criterion is applied based on a \textit{minimum congestion principle}: the seller selects the intersection with the smallest predicted number of participants. This mechanism not only preserves the objective of maximizing ESW but also mitigates potential resource congestion by promoting a more balanced spatial distribution of sellers.

Once all sellers complete the above steps, Algorithm \ref{alg4_2} outputs the global location deployment plan $Loc^{t}$. Since the decision process for each seller depends solely on local information, this module naturally supports distributed parallel computing. The corresponding computational complexity\footnote{The computational complexity is derived as follows: The algorithm iterates through each seller in $\bm{\mathcal{S}}$. For every seller, it explores a fixed neighborhood of (i.e., 5) candidate locations (the current position plus four adjacent directions). The core operation at each candidate location involves invoking SW calculation subroutine (denoted as $T_{\text{Alg\ref{alg4_3}}}$) to determine the marginal contribution. Therefore, the total time cost is the product of the number of sellers, the constant search space size (i.e., 5), and the runtime of a single subroutine execution.} is $O(|\bm{\mathcal{S}}| \times 5 \times T_{\text{Alg\ref{alg4_3}}})$, where $T_{\text{Alg\ref{alg4_3}}}$ reflects the complexity of ESW calculation.

\begin{algorithm}[]
	{\footnotesize
		\caption{Module A: UAV Trajectory Planning for the Next Timeslot}
		\label{alg4_2}
		\LinesNumbered
		
		{\bf{Input :}} 
		$\bm{\mathcal{S}}$: Set of all UAVs, $Loc^{t-1}$: Set of UAV coordinates at timeslot $t-1$, Set of all buyer parameters
		
		{\bf{Output :}} 
		$Loc^{t}$: Set of UAV coordinates at timeslot $t$
		
		{\bf{Initialization :}} 
		$Cand \leftarrow \emptyset$
		
		\For{all UAV $s_m \in \bm{\mathcal{S}}$}{
			$cur^{loc} \leftarrow Loc^{t-1}(s_m)$\
			
			$N^{locs} \leftarrow \{cur^{loc}\} \cup \text{GetAdjacentLocs}(cur^{loc})$\
			
			$MC \leftarrow \emptyset$\
			
			\For{all candidate intersection $l \in N^{locs}$}{
				Predicted member set $\bm{\mathcal{S}}'_l \leftarrow \{s_k \mid s_k \text{ potentially moves to } l \} \setminus \{s_m\}$\
				
				Calculate baseline ESW $W_{base} \leftarrow \text{Algorithm \ref{alg4_3}}(\bm{\mathcal{S}}'_l)$\
				
				Calculate new contribution $W_{new} \leftarrow \text{Algorithm \ref{alg4_3}}(\bm{\mathcal{S}}'_l \cup \{s_m\})$\
				
				Marginal contribution $\text{Mar}^\text{Con} \leftarrow W_{new} - W_{base}$\
				
				$MC \leftarrow MC \cup \{(l, \text{Mar}^\text{Con}, |\bm{\mathcal{S}}'_l|)\}$\
			}
			
			\textbf{\# Filter optimal candidates:}\
			
			$max\_\text{Mar}^\text{Con} \leftarrow \max\{\text{Mar}^\text{Con} \mid (l, \text{Mar}^\text{Con}, n) \in MC\}$\
			
			$Best \leftarrow \{(l, n) \mid (l, \text{Mar}^\text{Con}, n) \in MC \text{ and } \text{Mar}^\text{Con} = max\_\text{Mar}^\text{Con}\}$\
			
			Select $target\_loc \leftarrow \arg\min_{(l,n)\in Best} n$\
			
			Update $Loc^{t}(s_m) \leftarrow target\_loc$
		}
		
		\textbf{Return} $Loc^{t}$
	}
\end{algorithm}

\subsubsection{Module B: Double Auction for Timeslot $t$ (Buyer--Seller Matching)}

Given the sets of buyers and sellers assigned to each intersection by Module A, Module B (Algorithm \ref{alg4_3}) performs the core matching process within the auction mechanism, namely, determining the winning buyer--seller pairs. This module comprises four steps:

\noindent$\bullet$ \textit{Step 1: Dynamic state update and trajectory obfuscation (lines 3-10).} 
Before entering the core auction process, each buyer updates its internal state. Specifically, the privacy budget $\xi_n^t$ is adjusted according to the proposed feedback-driven stochastic model (line 6), capturing the buyer's current privacy--utility preference. In parallel, the demand state $\bm{\mathcal{Q}}_n^t$ is updated based on historical data (line 7). To protect privacy, the buyer also generates a virtual trajectory $\hat{\bm{\mathcal{L}}}_n^{\mathsf{b},t}$ using location obfuscation model detailed in Section \ref{sec_location_obfuscation} (line 10). This virtual path serves as the basis for following buyer-seller matching process.

\noindent$\bullet$ \textit{Step 2: Similarity-based clustering (lines 11-14).} 
Directly matching buyers and sellers across multiple service types incurs significant computational complexity due to its combinatorial nature. To improve tractability, this step partitions participants into clusters based on their spatio-temporal compatibility. Specifically, we compute the trajectory similarity $\Gamma_{m,n}^t$ between each buyer's virtual trajectory and each seller's proactive trajectory. Based on these similarity scores, a density-peak-based clustering method is applied to group highly compatible participants into clusters $\{\bm{\mathcal{G}}_k\}$ (line 14). This procedure effectively localizes the matching process by restricting it to smaller, high-potential clusters, thereby reducing computational overhead while preserving matching quality.

\noindent$\bullet$ \textit{Step 3: Feasibility screening and preference list construction (lines 15-25).} 
The matching process is performed independently for each service type $j$, such that the fundamental unit of allocation is a buyer--seller--service triplet. For a given service type $j$ within each cluster, buyers are first sorted in descending order of their bids, while sellers are sorted in ascending order of their asking prices. A VCG-inspired screening procedure (Lines 22--25) is then applied to identify the ``key indices'' ($k_b^*, k_s^*$) at which the bid--ask feasibility condition is satisfied (i.e., the marginal buyer's bid covers the marginal seller's ask). All sellers falling within this feasible trading range are added to the buyer's \textit{service-specific} preference list $\bm{\mathcal{O}}_{k_b}^{j}$ (line 25). This list serves a dual purpose: it defines the set of admissible candidates for contract formation in Phase 1, while simultaneously preserving ranked backup options that can be activated in Phase 2 for contingency handling.

\noindent$\bullet$ \textit{Step 4: ESW-driven matching (lines 26-35).} 
The objective of this step is to construct the final matching matrix by selecting winning buyer--seller pairs. All feasible buyer--seller--service triplets are first ranked in descending order according to their marginal contributions to the ESW (line 27). The algorithm iteratively selects pairs from this ordered list, assigning a matching only if neither the buyer nor the seller has been previously matched for the same service type (lines 29–35). Once a seller is assigned to a buyer as the primary contract partner, it is removed from the corresponding buyer's preference list (line 32), thereby distinguishing finalized matches from backup candidates. The process continues until no further feasible assignments remain, yielding the final total ESW $\overline{U^{\mathsf{sw},t}}$ and the final matching matrix $\bm{\mathcal{M}}^t$.

Upon completion, Module B provides confirmed matching relationships. These results are then passed to Module C, where the terms of the corresponding OSAAs are determined to ensure the desired economic properties.

\begin{algorithm}[]
	{\footnotesize
		\caption{Module B: Buyer--Seller Matching for Timeslot $t$ (Single-Intersection Auction)}
		\label{alg4_3}
		\LinesNumbered
		
		{\bf{Input :}} 
		$\bm{\mathcal{B}}^t$: Buyer set at timeslot $t$, $\bm{\mathcal{S}}^t$: Seller set at timeslot $t$, $\bm{\mathcal{L}}_n^{\mathsf{b},t-1}$: Future path of buyer $n$ at $t-1$, $\xi_n^{t-1}$: Privacy budget of buyer $n$ at $t-1$
		
		{\bf{Output :}} 
		$\bm{\mathcal{M}}^t$: Matching result matrix, $\overline{U^{\mathsf{sw},t}}$: Expected SW, $\bm{\mathcal{O}}_{k_b}^{j}$: Buyer preference list
		
		{\bf{Initialization :}} 
		$\bm{\mathcal{M}}^t \leftarrow \emptyset$, $\overline{U^{\mathsf{sw},t}} \leftarrow 0$
		
		\textbf{\# Procedure: Buyer Parameter Update}\\
		\For{each buyer $b_n \in \bm{\mathcal{B}}^t$}{
			$\xi_n^t \leftarrow \xi_n^{t-1} + \Delta \xi_n$ \tcp*{Update privacy budget via Markov process}
			$\bm{\mathcal{Q}}_n^t \leftarrow \text{UpdateDemand}(\bm{\mathcal{Q}}_n^{t-1}, \bm{\mathcal{L}}_n^{\mathsf{b},t-1})$ \tcp*{Update dynamic demand model}
		}
		
		\textbf{\# Procedure: Path Obfuscation}\\
		\For{each buyer $b_n \in \bm{\mathcal{B}}^t$}{
			$\hat{\bm{\mathcal{L}}}_n^{\mathsf{b},t} \leftarrow \text{GenerateObfuscatedPath}(\bm{\mathcal{L}}_n^{\mathsf{b},t}, \xi_n^t)$ \tcp*{Three-stage location obfuscation model}
		}
		
		\textbf{\# Procedure: Member Grouping}\\
		\For{$(b_n, s_m) \in \bm{\mathcal{B}}^t \times \bm{\mathcal{S}}^t$}{
			$\Gamma_{m,n}^t \leftarrow \frac{\hat{\bm{\mathcal{L}}}_n^{\mathsf{b},t} \cap \bm{\mathcal{L}}_m^{\mathsf{s},t}}{|\hat{\bm{\mathcal{L}}}_n^{\mathsf{b},t}|}$ \tcp*{Calculate path similarity}
		}
		$\{\bm{\mathcal{G}}_k\} \leftarrow \text{Cluster}(\Gamma^t, \delta_{\Gamma})$ \tcp*{$\delta_{\Gamma}$ is similarity threshold}
		
		\textbf{\# Procedure: VCG Matching}\\
		\For{each group $\bm{\mathcal{G}}_k \in \{\bm{\mathcal{G}}_k\}$}{
			\For{each service type $j=1$ \KwTo $J$}{
				Partition $\bm{\mathcal{B}}_{j,k},\bm{\mathcal{S}}_{j,k} \subseteq \bm{\mathcal{G}}_k$\
				
				\textbf{Generate sorted lists:}\\
				$\mathcal{L}_b \leftarrow \langle b_n: bid_{n}^{j} \downarrow \rangle$, $\mathcal{L}_s \leftarrow \langle s_m: ask_{m}^{j} \uparrow \rangle$\
				
				\textbf{Determine key indices:}\\
				\For{$k_b=1$ \KwTo $|\mathcal{L}_b|$}{
					\For{$k_s=1$ \KwTo $|\mathcal{L}_s|$}{
						\If{$bid_{k_b} \ge ask_{k_s}$ and boundary conditions met}{
							Record $(k_b^*, k_s^*)$; Generate preference list $\bm{\mathcal{O}}_{{k_b}}^{j}\leftarrow\mathcal{L}_s $\
						}
					}
				}
				
				$\bm{\mathcal{K}}_B \leftarrow \emptyset, \bm{\mathcal{K}}_S \leftarrow \emptyset$\
				
				\textbf{Generate SW sorted list:} $\mathcal{L}_{sw} \leftarrow \langle (k_b,k_s): SW_{k_b,k_s}^{j} \downarrow \rangle$\
				
				\textbf{Execute non-repetitive matching:}\\
				\For{$(k_b,k_s) \in \mathcal{L}_{sw}$}{
					\If{$k_b \notin \bm{\mathcal{K}}_B$ and $k_s \notin \bm{\mathcal{K}}_S$}{
						$\bm{\mathcal{M}}^t[k_b,k_s] \leftarrow 1$\\
						$\bm{\mathcal{O}}_{{k_b}}^{j} \leftarrow \bm{\mathcal{O}}_{{k_b}}^{j} \setminus k_s$\\
						$\bm{\mathcal{K}}_B \leftarrow \bm{\mathcal{K}}_B \cup \{k_b\}$\\
						$\bm{\mathcal{K}}_S \leftarrow \bm{\mathcal{K}}_S \cup \{k_s\}$\\
						$\overline{U^{\mathsf{sw},t}} \leftarrow \overline{U^{\mathsf{sw},t}} + \mathbbm{q}_{n}^{j,t}\left(\Gamma_{m,n}^{t}v_{n}^{j}-k_{n}^{j} \xi_{n}^{t}-c_{m}^{j}\right)$\
					}
				}
			}
		}
		
		\textbf{Return} $\bm{\mathcal{M}}^t$, $\overline{U^{\mathsf{sw},t}}$, $\{\bm{\mathcal{O}}_{k_b}^{j}\}$
	}
\end{algorithm}

\subsubsection{Module C: OSAA Term Determination and Pricing for Timeslot $t$}

To ensure key economic properties (i.e., motivating participants to report true costs and valuations), Module C (Algorithm \ref{alg4_4}) borrows idea from the VCG-based pricing \cite{11115158} to calculate specific payments for buyers and revenues for sellers through three steps:

\noindent$\bullet$ \textit{Step 1: ESW decomposition (lines 1-3).} 
To quantify the economic impact of individual transactions, we decompose the ESW into pairwise contributions. Specifically, for each matched buyer--seller--service triplet $(s_m, b_n, j)$ in the matching matrix $\bm{\mathcal{M}}^t$, the algorithm computes its marginal ESW contribution using (\ref{eq_4-13}), denoted by $\Delta U_{m,n}^{\mathsf{sw},j}$ (line 3). This term captures the net contribution of the corresponding transaction to the overall ESW, serving as the fundamental building block for subsequent pricing decisions.

\noindent$\bullet$ \textit{Step 2: Buyer payment determination (lines 7-13).} 
For each winning buyer $b_n$, the payment is computed based on the VCG externality principle, i.e., the impact of $b_n$'s participation on the welfare of other participants. Specifically, for each service type $j$, we construct a counterfactual scenario in which $b_n$ is excluded (i.e., setting $\mathbbm{q}_{n}^{j,t} = 0$ and defining the reduced buyer set $\bm{\mathcal{B}}'$). The matching process is then re-executed to obtain the corresponding counterfactual ESW, denoted by $\overline{U_{-n}^{\mathsf{sw}}}$. The final payment is determined by the difference between: \textit{(i)} the ESW achieved by all other participants in the counterfactual scenario, and \textit{(ii)} the ESW of those same participants in the actual outcome (line 12). This construction ensures that each buyer is charged according to the marginal externality it imposes on the system, thereby enforcing IC.

\noindent$\bullet$ \textit{Step 3: Seller revenue determination (lines 14-20).} 
Similarly, for each winning seller $s_m$, the revenue is computed based on its marginal contribution to the ESW. Specifically, for each service type $j$, we construct a counterfactual scenario in which $s_m$ is excluded (e.g., by removing it from the seller set or setting its ask to an effectively prohibitive value). The matching process is then re-executed to obtain the corresponding counterfactual ESW, denoted by $\overline{U_{-m}^{\mathsf{sw}}}$. The revenue is then determined by the difference between: \textit{(i)} the ESW achieved by all other participants in the actual outcome, and \textit{(ii)} the ESW of those same participants in the counterfactual scenario (line 19).

\begin{algorithm}[]
	{\footnotesize
		\caption{Module C: OSAA Term Determination and Pricing for Timeslot $t$}
		\label{alg4_4}
		\LinesNumbered
		
		{\bf{Input :}} 
		$\bm{\mathcal{M}}^t$: Matching matrix from Algorithm \ref{alg4_3}, $\bm{\mathcal{B}}^t$: Buyer set, $\bm{\mathcal{S}}^t$: Seller set, $\{\Gamma_{m,n}^t\}$: Path similarity, $\{\mathbbm{q}_n^{j,t}\}$: Demand intensity, $\{v_n^j\}$: Service valuation, $\{k_n^j\}$: Privacy cost coefficient, $\{c_m^j\}$: Operation cost, $\{\xi_n^t\}$: Privacy budget, $\overline{U^{\mathsf{sw},t}}$: Expected SW
		
		{\bf{Output :}} 
		$\bm{p}_{m,n}$: Buyer payment set, $\bm{r}_{m,n}$: Seller revenue set
		
		{\bf{Initialization :}} 
		$\bm{p}_{m,n} \leftarrow \emptyset$, $\bm{r}_{m,n} \leftarrow \emptyset$, $\Delta U_{m,n}^{\mathsf{sw},j} = x_{m,n}^{j,t} \cdot \mathbbm{q}_n^{j,t} \cdot (\Gamma_{m,n}^t v_n^j -k_n^j\xi_n^t -c_m^j)$
		
		\textbf{\# Procedure: Exclusive Matching ($\bm{\mathcal{B}}', \bm{\mathcal{S}}'$)}\\
		Replicate Algorithm \ref{alg4_3} steps: Execute buyer parameter update, path obfuscation, member grouping, VCG matching\\
		Output: $\bm{\mathcal{M}}'$, $\overline{U^{\mathsf{sw}}}$ \tcp*{New matching based on $\bm{\mathcal{B}}', \bm{\mathcal{S}}'$}
		
		\textbf{\# Procedure: Buyer Pricing}\\
		\For{each service type $j=1$ \KwTo $J$}{
			\For{each buyer $b_n \in \bm{\mathcal{B}}^t$}{
				Construct exclusion set: $\bm{\mathcal{B}}' \leftarrow \bm{\mathcal{B}}^t (\mathbbm{q}_{n}^{j,t} \leftarrow 0)$\\
				Call Exclusive Matching: $(\bm{\mathcal{M}}'_n, \overline{U_{-n}^{\mathsf{sw}}}) \leftarrow \text{Exclusive Matching}(\bm{\mathcal{B}}', \bm{\mathcal{S}}^t)$\\
				Determine VCG payment: $p_{m,n}^{j,t} \leftarrow \overline{U_{-n}^{\mathsf{sw}}} - (\overline{U^{\mathsf{sw},t}} - \Delta U_{m,n}^{\mathsf{sw},j})$ \tcp*{$\Delta U_{m,n}^{\mathsf{sw},j}$ is original matching contribution}
				$\bm{p}_{m,n} \leftarrow \bm{p}_{m,n} \cup \{p_{m,n}^{j,t}\}$\
			}
		}
		
		\textbf{\# Procedure: Seller Pricing}\\
		\For{each service type $j=1$ \KwTo $J$}{
			\For{each seller $s_m \in \bm{\mathcal{S}}^t$}{
				Construct exclusion set: $\bm{\mathcal{S}}' \leftarrow \bm{\mathcal{S}}^t(ask_{m}^{j} \leftarrow \text{max}[ask^{j}])$\\
				Call Exclusive Matching: $(\bm{\mathcal{M}}'_m, \overline{U_{-m}^{\mathsf{sw}}}) \leftarrow \text{Exclusive Matching}(\bm{\mathcal{B}}^t, \bm{\mathcal{S}}')$\\
				Calculate VCG revenue: $r_{m,n}^{j,t}  \leftarrow \overline{U^{\mathsf{sw},t}} - \overline{U_{-m}^{\mathsf{sw}}}$\\
				$\bm{r}_{m,n} \leftarrow \bm{r}_{m,n} \cup \{r_{m,n}^{j,t} \}$\
			}
		}
		
		\textbf{Return} $\bm{p}_{m,n}$, $\bm{r}_{m,n}$
	}
\end{algorithm}

\subsection{Real-time Execution of OSAAs and Preference Lists during Phase 2}

Upon the completion of Phase 1 of LOSA, the OSAAs are finalized, specifying the service types, delivery locations, and transaction prices. In the subsequent Phase 2 of LOSA, as buyers and sellers physically arrive at the designated intersections, these pre-established agreements are executed.

In the presence of real-time uncertainties, such as unforeseen demand or contract failures, the system leverages the preference lists generated in Module B of Phase 1 to enable rapid supplementary matching. This mechanism ensures service continuity without requiring a full re-execution of the auction. Since the heavy decision-making lifting (matching and pricing) has been “pre-computed”, Phase 2 of LOSA focuses on rapid execution and lightweight adjustment, involving the following steps:

\noindent$\bullet$ \textit{Step 1: On-site verification and OSAA execution.} 
Upon arrival, the system first verifies the actual realization of buyer's stochastic demand $Q_{n}^{j,t}$. If the demand is active (i.e., $Q_{n}^{j,t}=1$), the pre-signed OSAA is immediately executed, and the service is delivered according to $\mathbb{C}_{m,n}^{t}$. Otherwise, the reserved resources of the seller are released for further potential reallocation.

\noindent$\bullet$ \textit{Step 2: Preference-based contingency matching.} 
If a buyer’s demand remains unmet, either due to failing to secure a contract in Phase 1 or an on-site low-probability demand realization, it does not need to wait for a new auction. Instead, the buyer utilizes its preference list $\bm{\mathcal{O}}_{n}^{t}$ generated in Module B of Phase 1 to find backup sellers. to identify feasible backup sellers. The highest-ranked available seller is then selected for immediate service provisioning, ensuring service provisioning continuity with minimal latency.

\noindent$\bullet$ \textit{Step 3: Conflict resolution.} 
Since preference lists are constructed independently, multiple buyers may select the same seller as a fallback, potentially violating the one-to-one matching constraint for each service type. To resolve such conflicts without global recomputation, we adopt a greedy tie-breaking rule: among competing buyer--seller pairs, the one yielding the highest marginal ESW contribution is retained, while the others are discarded or deferred.

By integrating both phases, LOSA is theoretically guaranteed to satisfy key economic properties, including IR, IC, CE, and BB, across both the proactive decision-making and real-time execution phases. Rigorous mathematical proofs establishing these properties are provided in Appendix \ref{adx:4}.

\section{Evaluations}
\label{sec:evaluations}

We next conduct a set of experiments (using MATLAB R2024b) to assess the performance of LOSA. To construct an evaluation environment that captures diverse real-world scenarios, including urban road networks, intersection blind spots, and highway settings, we integrate three datasets: DAIR-V2X\footnote{DAIR-V2X Dataset, Tsinghua University. \url{https://air.tsinghua.edu.cn/DAIR-V2X}}, HighD\footnote{HighD Dataset Official Website. \url{https://www.highd-dataset.com}}, and RCooper\footnote{RCooper Open Source Repository. \url{https://github.com/AIR-THU/DAIR-RCooper}}.

\subsection{Constructing the Simulation Environment from Real-World Datasets}

We employ spatio-temporal alignment techniques to integrate multiple heterogeneous datasets into a unified framework. Each dataset contributes a complementary dimension (mobility, demand dynamics, and sensing heterogeneity), enabling a holistic evaluation of the proposed framework. The datasets we consider are as follows:
\textit{(i) DAIR-V2X.} The DAIR-V2X dataset provides cooperative perception data over 10 km of urban roads in the Beijing demonstration zone, including vehicle trajectories with tracking IDs. We extract buyer and seller trajectories, denoted as $\bm{\mathcal{L}}_{n}^{\mathsf{b},t}$ and $\bm{\mathcal{L}}_{m}^{\mathsf{s},t}$, respectively, by segmenting raw traces into 10-second clips with a timeslot granularity of $\Delta t = 100$ ms. GPS coordinates are discretized via bilinear interpolation into a 5000 m×5000 m Manhattan grid with a resolution of 200 m, which serves as the unit for the privacy radius $\mathbbm{r}_n$. \textit{(ii) HighD.} The HighD dataset contains lane-changing behaviors of over 110,000 vehicles on German highways. We use its car-following statistics to parameterize the demand triggering probability $\mathbbm{q}_{n}^{j,t}$, modeled as an exponential decay function $\mathbbm{q}_{n}^{j,t} = \lambda_0 e^{-0.03t}$, whose values typically fall within [0.7, 0.95]. \textit{(iii) RCooper.} The RCooper intersection dataset includes 50,000 blind-spot images and 30,000 point-cloud frames. We leverage this dataset to model the heterogeneity in seller cost vectors $\bm{\mathcal{C}}_m$, particularly reflecting sensing, computation, and communication overheads. Additionally, spatial clustering is applied to identify high-demand regions, supporting proactive UAV placement and trajectory planning. Due to space constraints, detailed preprocessing procedures (e.g., trajectory segmentation and interpolation), demand modeling, and the complete set of simulation parameters are provided in Appendix \ref{appendix:exp_details}.

Regarding privacy protection, our parameter settings follows the probabilistic perturbation framework \cite{zhang2024trajectory} and Path Perturbation Theory \cite{10815979}. The privacy radius $\mathbbm{r}_n$ is set to 3-5 grid units (corresponding to a physical distance of 30-50 meters). The discretization steps $\Delta r_n$ and $\Delta \theta_n$ satisfy the proportional constraint $\Delta r_n \propto 1/\mathbbm{r}_n$ (as detailed in Table \ref{tab:sim_params_appendix}). Privacy budget $\xi_{n}^{t}$ follows a dynamic adjustment mechanism with an initial value of $\xi_0=2.5$ and an upper bound $\xi_{\max}=5$, using the proposed update rule in (\ref{eq3}). Finally, the trajectory similarity threshold $\delta_\Gamma$ ($\delta_{\Gamma}$ is the similarity threshold used in Algorithm \ref{alg4_3}) is normalized via the Fréchet distance to ensure that service quality satisfies $\bm{\mathcal{Q}}_{n}^{t} \geq 0.8$ after trajectory obfuscation.

For service transaction parameters, we adopt configurations consistent with prior studies \cite{11204540}.
Specifically, the buyer's service valuation vector $\bm{\mathcal{V}}_n$ and the privacy cost coefficient vector $\bm{\mathcal{K}}_n$ are calibrated under VCG mechanism. We set the number of service types\footnote{According to the IEEE 1609.4 WAVE standard, the 5.9 GHz spectrum provides one Control Channel (CCH) and six Service Channels (SCHs). Since Channel 172 is typically strictly reserved for critical vehicle-to-vehicle (V2V) safety communications, we assign the remaining five SCHs to support the general service provisions.} to $J=5$. For sellers, the service cost vector $\bm{\mathcal{C}}_m$ is calculated using the energy consumption model of RSUs \cite{8428533}.
We also incorporate sensor heterogeneity from the RCooper dataset to validate the tiered pricing strategy.

\subsection{Evaluation Metrics and Benchmark Methods}
\label{subsec:metrics and baselines}

To better evaluate LOSA, four core metrics are used:

\noindent $\bullet$ \textit{Social welfare (SW)}:
SW represents the primary metric of market efficiency, capturing the net utility of the system computed as in (\ref{eq20}).

\noindent $\bullet$ \textit{Time consumed by auction decision-making (TimeADM):}
TimeADM measures the computation time required to produce the final service delivery, which can reflect the time-efficiency gains from trajectory clustering and dynamic privacy-budget adjustment in reducing computational complexity and mitigating bidding conflicts.

\noindent $\bullet$ \textit{Buyer utility (BU):}
BU measures a buyer's net utility, depending on service valuation, payment, and privacy cost computed as in (\ref{eq7}).

\noindent $\bullet$ \textit{Inference error (IE):}
IE serves as the primary metric for evaluating the effectiveness of trajectory privacy preservation. It quantifies the Euclidean distance between an adversary's estimated location (derived from Bayesian inference attacks on the obfuscated virtual trajectory \cite{8634952}) and the buyer's true location.


Then, the following benchmark methods are considered (more details are provided in Appendix \ref{appendix:baseline_details}). 

\noindent $\bullet$ \textit{VCG-based real-time Auction (VRA) \cite{11431960}:}
VRA adopts a single-phase VCG-based real-time auction with proactive UAV path planning, where auctions are triggered only upon intersection arrival and rely solely on instantaneous supply/demand. Compared with LOSA, it highlights the advantage of our two-phase design in ensuring service continuity and handling unexpected demand surges.

\noindent $\bullet$ \textit{Static VCG-based real-time auction (SVRA) :} 
SVRA removes UAV path planning from VRA and enforces fixed UAV trajectories. Auctions are conducted based on static location matching, without pre-negotiation or adaptive coverage. Its comparison with LOSA demonstrates the importance of active path planning on improving SW over dynamism.

\noindent $\bullet$ \textit{No privacy-protection Auction (NPPA) \cite{11073624}:} 
NPPA removes all privacy protection while retaining the two-phase architecture and UAV path planning of LOSA. Buyers submit true trajectories directly, enabling the evaluation of the utility loss introduced by our privacy protection method.

\noindent $\bullet$ \textit{Fixed high privacy budget (FHPB) :} 
FHPB fixes the privacy budget at a high level, yielding high matching accuracy at the cost of increased privacy leakage, helping highlight the benefit of dynamic privacy budget regulation.

\noindent $\bullet$ \textit{Fixed low privacy budget (FLPB) :} 
FLPB enforces the maximum privacy at the cost of matching accuracy. Corresponding comparison tests system robustness under extreme privacy constraints and the ability of dynamic regulation to avoid privacy deadlock.

These form a benchmark suite along three dimensions: mechanism architecture (VRA/SVRA), privacy module (NPPA), and budget dynamics (FHPB/FLPB). Collectively, they enable a systematic ablation-style assessment of LOSA and its key design contributions.

\subsection{Performance Evaluation}
\label{subsec:results}

\noindent $\bullet$ \textit{Analysis of SW.} Fig. \ref{fig:social_welfare} shows the comparison of SW under varying problem scales. We construct 16 different combinations with the number of buyers $N \in \{50, 100, 150, 200\}$ and the number of sellers $M \in \{20, 30, 40, 50\}$ to verify the performance of LOSA.

\textit{At a small transaction scale ($N \leq 100, M \leq 30$)}, VRA yields the highest SW because it performs instantaneous matching on unperturbed real-time information. 
For $N=50$ and $M=20$, LOSA achieves $92.4\%$ of VRA’s SW while substantially improving stability, evidencing that the look-ahead phase reduces decision risk and dampens demand-driven volatility.

\textit{As the scale increases ($N \geq 150, M \geq 40$)}, both VRA and SVRA exhibit pronounced performance degradation: their SW drops by $41.2\%$ and $63.8$, respectively. This is because their decision latency exceeds the maximum allowable intersection time window ($Time_{\max}$), which triggers large-scale transaction timeouts. 
In contrast, LOSA sustains continuous SW growth by leveraging path-similarity clustering and elastic privacy-budget adaptation, outperforming NPPA by $17.0\%$, and demonstrating strong scalability under high concurrency. \textit{For fixed-budget privacy schemes}, FHPB ($\xi_{\max}=5$) exposes excessive location information, and the resulting privacy cost substantially erodes utility, leading to an SW that is $12.4\%$ lower than LOSA at $N=200$ and $M=50$. Conversely, FLPB ($\xi_{\min}=1$) introduces overly aggressive trajectory obfuscation, reducing matching accuracy and effective transaction volume. Overall, the robust SW achieved by LOSA at $N=200$ and $M=50$ confirms its ability to complete high-concurrency trading within the $Time_{\max}$ constraint, whereas other baselines fail to satisfy the stringent latency requirement.

\begin{figure}[]  
	\centering  
	\subfigure[] {
		\label{fig:social_welfare}
		\includegraphics[trim=0cm 0cm 0cm 0cm, clip,width=0.2\textwidth,height=0.18\textwidth]{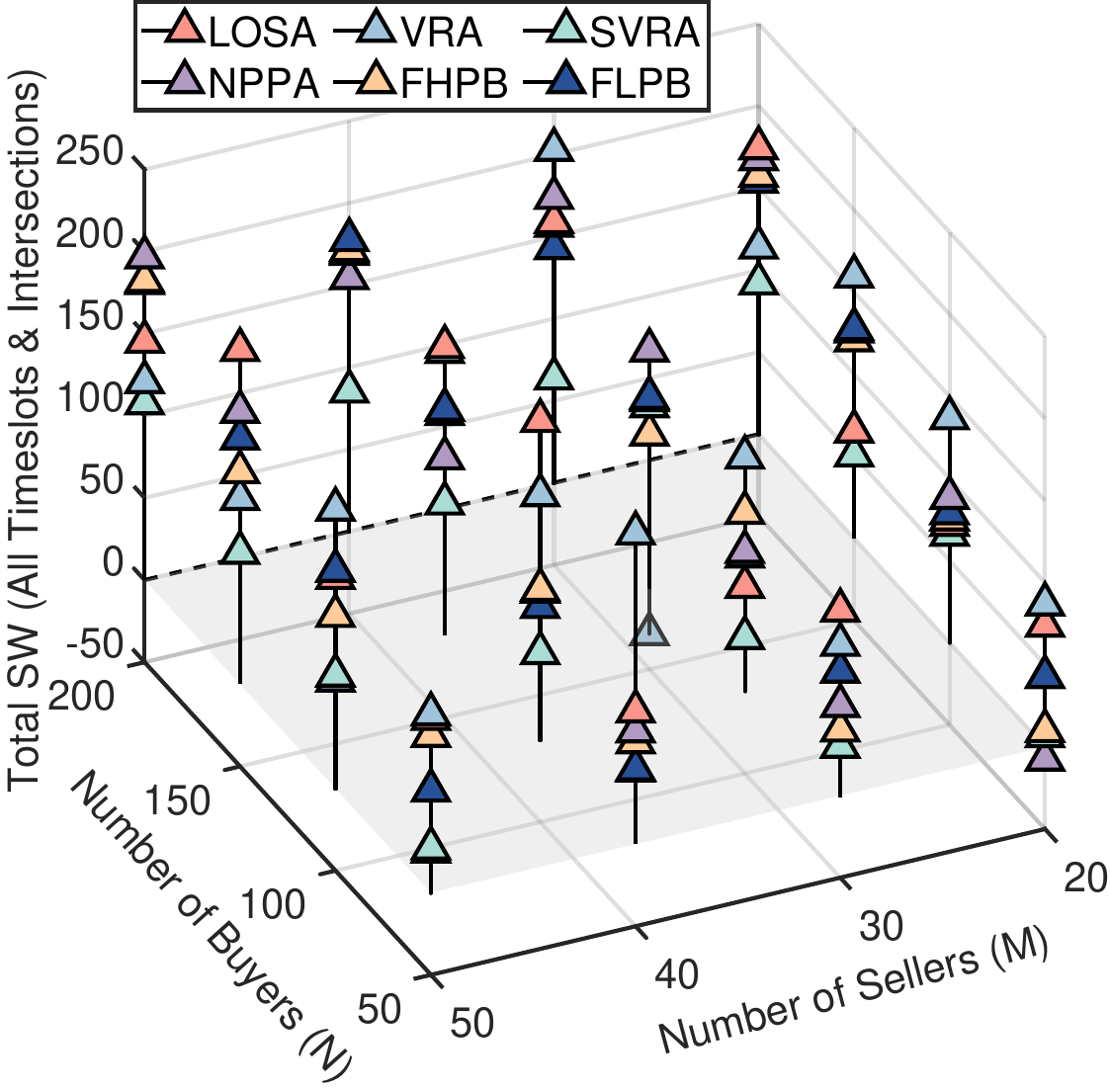}}  
	\subfigure[] {
		\label{fig:time_comparison}
		\includegraphics[trim=0cm 0cm 0cm 0cm, clip,width=0.2\textwidth,height=0.18\textwidth]{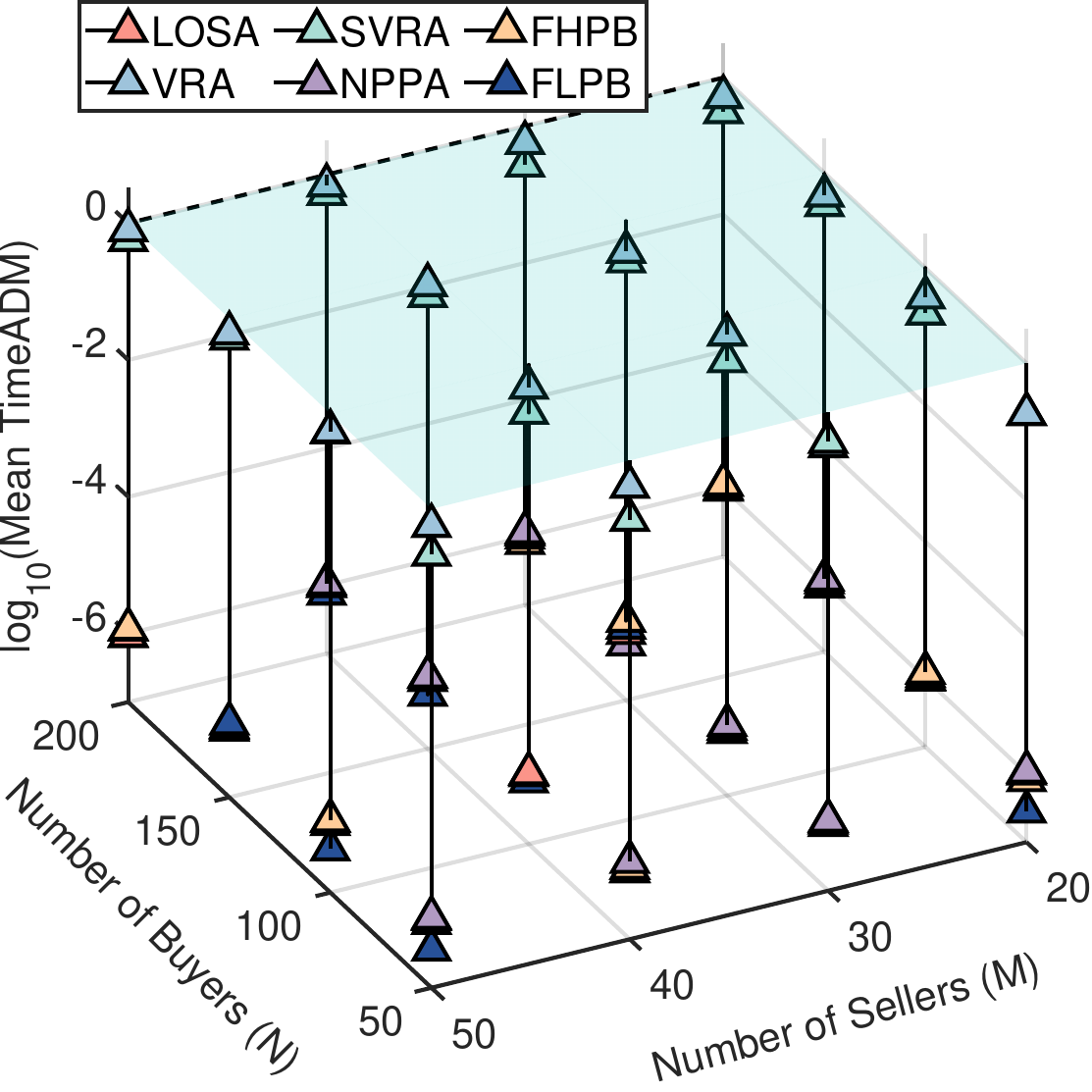}}  
	\caption{Performance comparison on SW and TimeADM.}  
	\label{fig4_3}  
\end{figure}

\noindent $\bullet$ \textit{Analysis of time efficiency.} Fig. \ref{fig:time_comparison} reports the logarithmic decision time, i.e., $\log_{10}(\text{TimeADM})$, averaged over the entire simulation period (all $T=100$ timeslots). We evaluate 16 buyer-seller combinations ($N \in \{50, 100, 150, 200\}, M \in \{20, 30, 40, 50\}$) to validate the timeliness across different methods under a hard intersection deadline of $Time_{\max}=1\text{s}$. While the specific time values depend on the experimental hardware (AMD Ryzen 97945HX), the relative orders of magnitude demonstrate the intrinsic complexity differences.

\textit{When the transaction scale is small ($N \leq 100, M \leq 30$),} the mean $\log_{10}(\text{TimeADM})$ for LOSA and other planning-based variants (NPPA, FHPB, FLPB) remains stable at approximately $-6.0$ (corresponding to $\approx 1\mu \text{s}$). This indicates that the trajectory similarity clustering in Phase 1 of LOSA (with $\delta_\Gamma \geq 0.85$) effectively reduces the bidding conflict rate in Phase 2 of LOSA. In contrast, VRA and SVRA are significantly slower, with $\log_{10}(\text{TimeADM}) \approx -1.0$ ($\approx 100\text{ms}$) and $-1.3$ ($\approx 50\text{ms}$), respectively. Notably, SVRA is slightly faster because its fixed UAV trajectories prune valid transaction nodes, though this comes at the cost of a reduction in SW as seen from Fig. \ref{fig:social_welfare}.

\textit{As the scale expands ($N \geq 150, M \geq 40$),} the runtime of VRA rises rapidly to $\log_{10}(\text{TimeADM}) = 0.0$ ($\approx 1\text{s}$), hitting the upper limit of the time window. At this point, SVRA reaches $-0.097$ ($\approx 800\text{ms}$). The runtime gap between VRA and SVRA validates that the proactive UAV path planning (i.e., Module A of Phase 1 of LOSA, which guides UAV mobility based on marginal contribution evaluation) increases trading opportunities but inevitably adds computational load. Among the planning-based algorithms, FHPB ($\xi_{\max}=5$) incurs the least perturbation and rises only to $-5.2$ ($\approx 0.6\mu \text{s}$). Compared to LOSA ($-5.5$ at extreme scale), this gap reflects the additional yet controlled overhead introduced by our dynamic budget adjustment.

\textit{Under a rather larger scale ($N=200, M=50$)}, based on the settlement results from the simulated transactions, the transaction failure rate of VRA reaches 82.7\% due to timeouts. In contrast, LOSA controls the time cost at $-5.5$ ($\approx 0.3\mu \text{s}$) by using an elastic privacy radius ($\mathbbm{r}_n \in [3, 5]$ grids), representing an optimization compared to the static strategy FLPB ($\xi_{\min}=1$). Overall, results demonstrate that our two-phase design leads to scalability more than three orders of magnitude better than traditional auctions.

\noindent $\bullet$ \textit{Analysis of BU.} 
To investigate the micro-level performance of individual participants, we analyze the evolution of BU from two distinct perspectives: the internal self-adaptation of parameters and the external impact of different privacy protection strategies.

\begin{figure*}[htbp]
	\centering
	\subfigure[] {\includegraphics[trim=0cm 0cm 0cm 0cm, clip, width=.165\textwidth,height=0.12\textwidth]{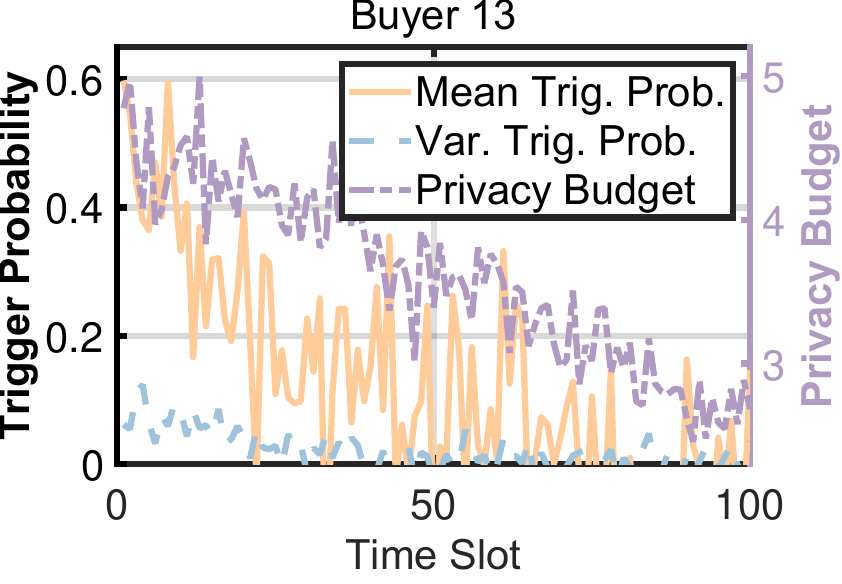}}
	\subfigure[] {\includegraphics[trim=0cm 0cm 0cm 0cm, clip, width=.159\textwidth,height=0.12\textwidth]{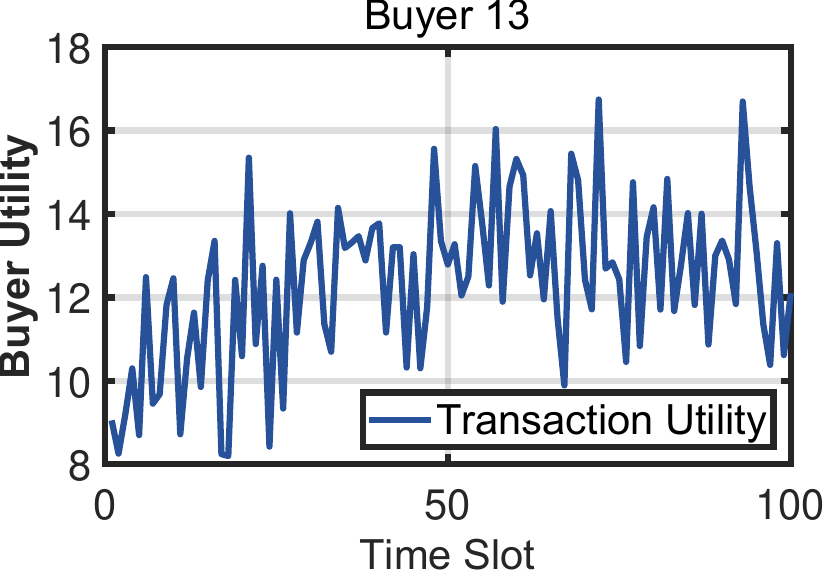}}
	\subfigure[] {\includegraphics[trim=0cm 0cm 0cm 0cm, clip, width=.165\textwidth,height=0.12\textwidth]{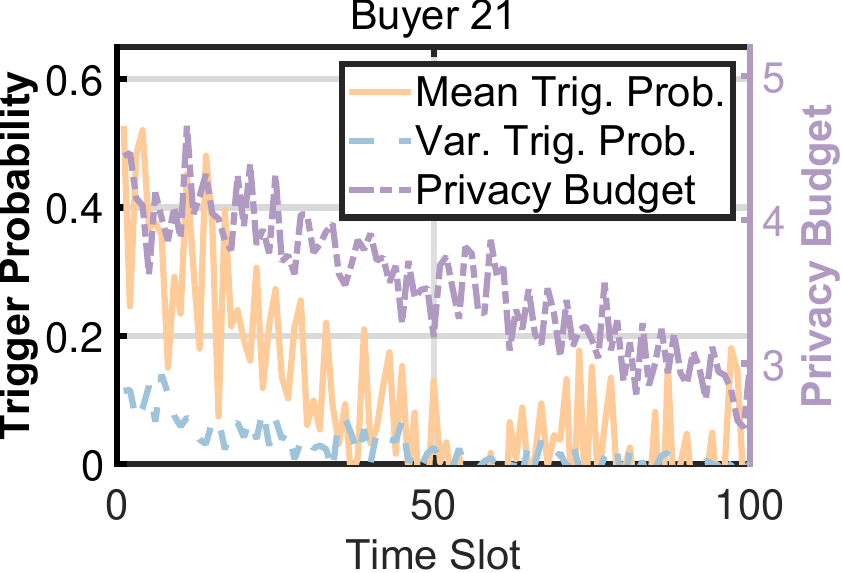}}
	\subfigure[] {\includegraphics[trim=0cm 0cm 0cm 0cm, clip, width=.159\textwidth,height=0.12\textwidth]{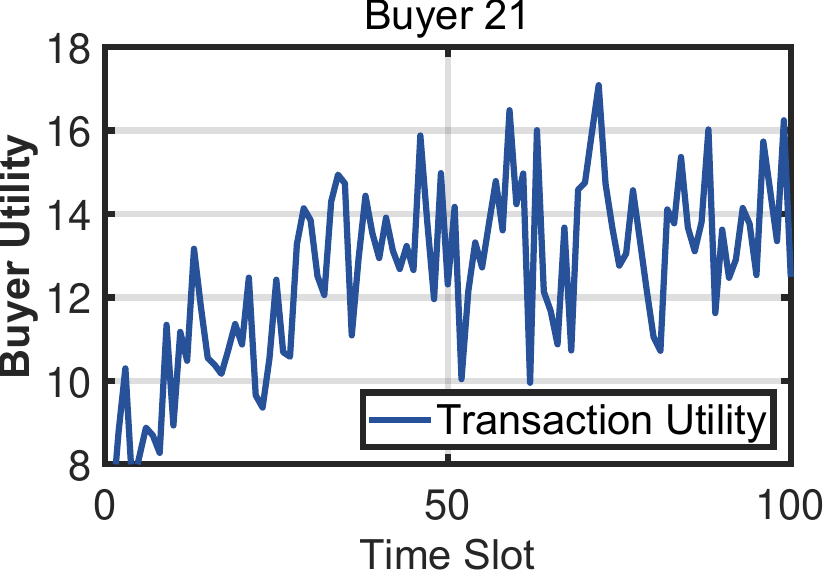}}
	\subfigure[] {\includegraphics[trim=0cm 0cm 0cm 0cm, clip, width=.165\textwidth,height=0.12\textwidth]{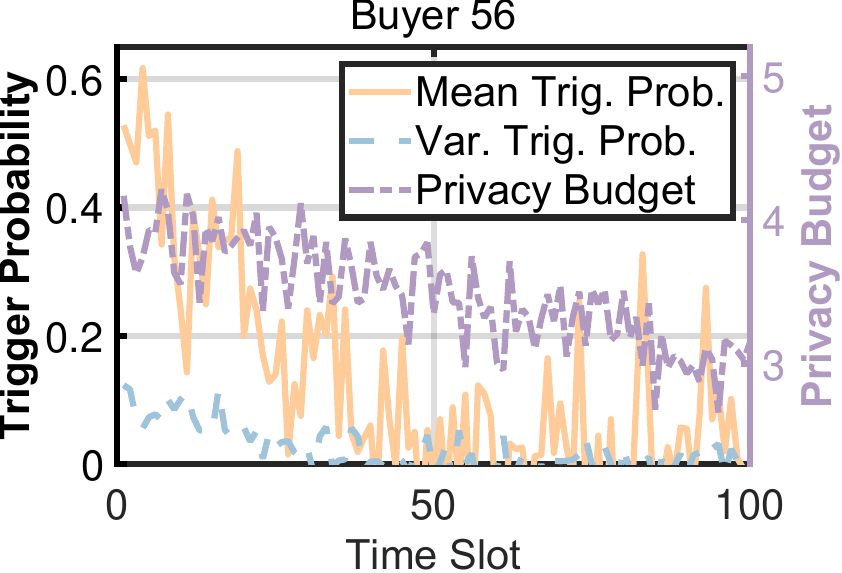}}
	\subfigure[] {\includegraphics[trim=0cm 0cm 0cm 0cm, clip, width=.159\textwidth,height=0.12\textwidth]{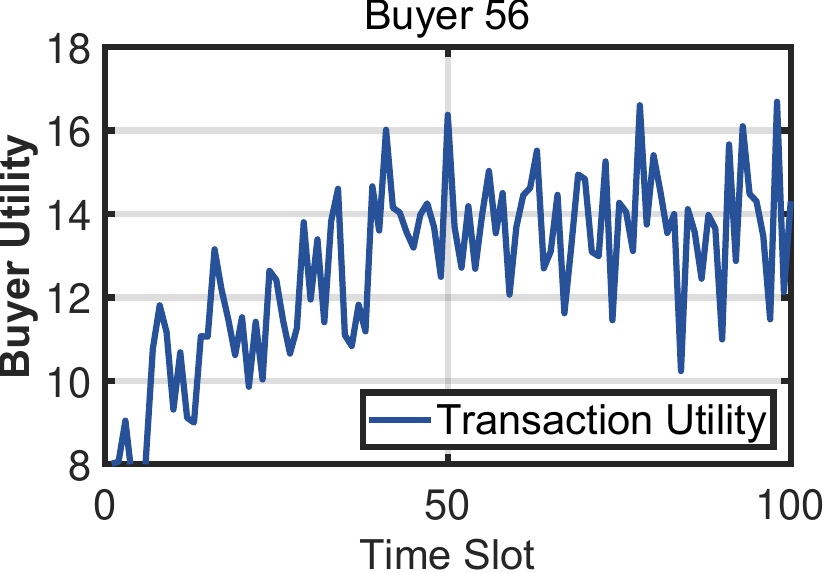}}
	\subfigure[] {\includegraphics[trim=0cm 0cm 0cm 0cm, clip, width=.165\textwidth,height=0.12\textwidth]{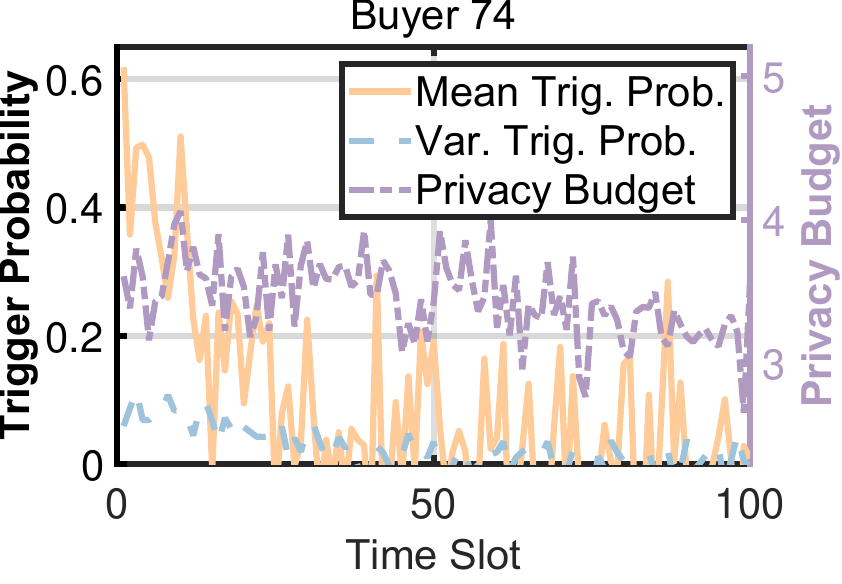}}
	\subfigure[] {\includegraphics[trim=0cm 0cm 0cm 0cm, clip, width=.159\textwidth,height=0.12\textwidth]{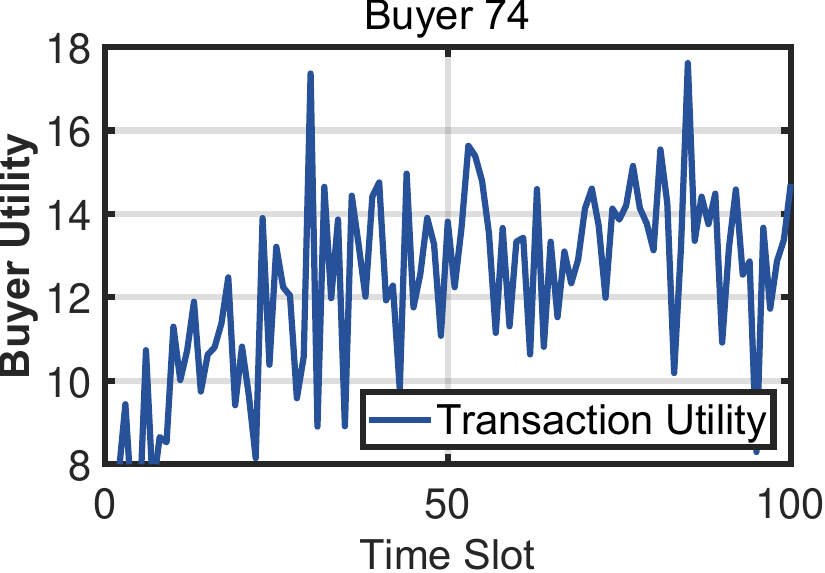}}
	\subfigure[] {\includegraphics[trim=0cm 0cm 0cm 0cm, clip, width=.162\textwidth,height=0.12\textwidth]{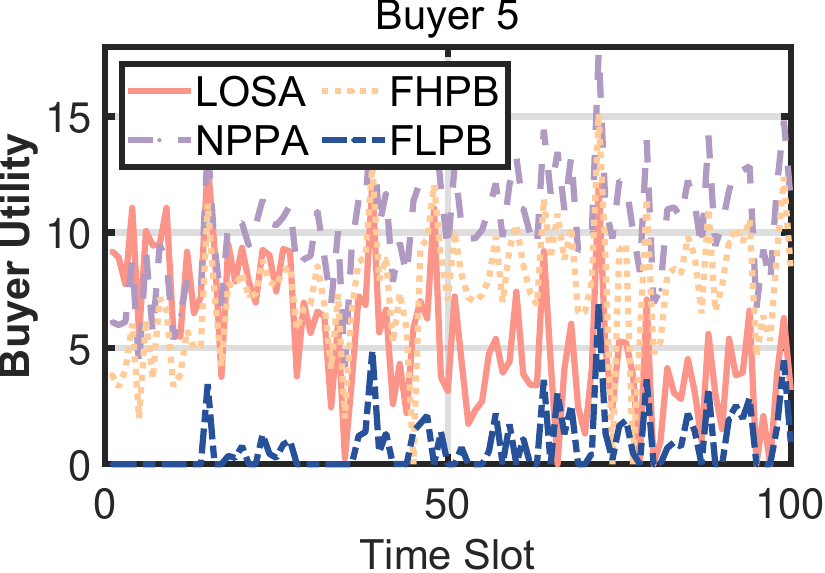}}
	\subfigure[] {\includegraphics[trim=0cm 0cm 0cm 0cm, clip, width=.162\textwidth,height=0.12\textwidth]{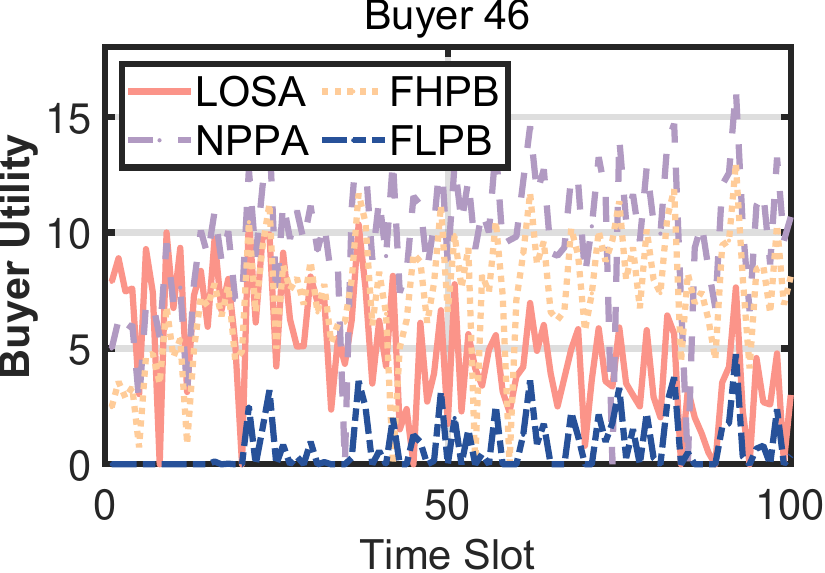}}
	\subfigure[] {\includegraphics[trim=0cm 0cm 0cm 0cm, clip, width=.162\textwidth,height=0.12\textwidth]{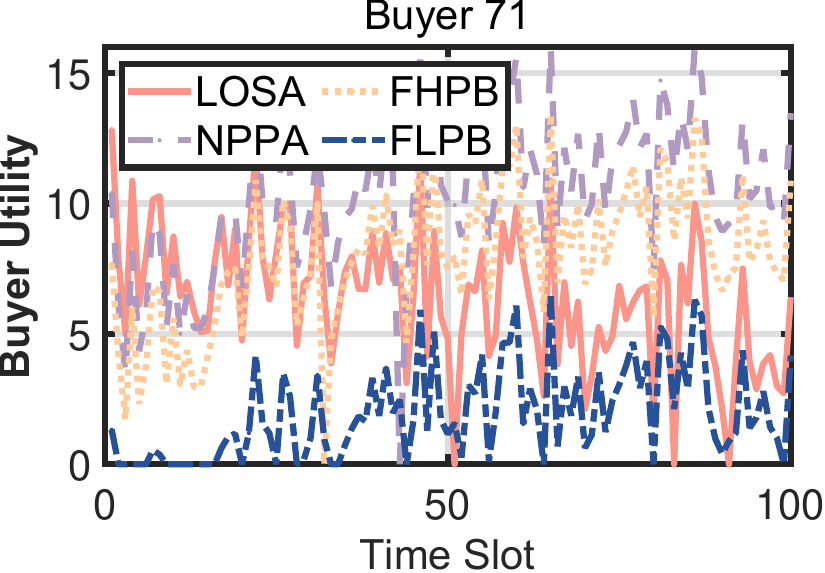}}
	\subfigure[] {\includegraphics[trim=0cm 0cm 0cm 0cm, clip,
		width=.162\textwidth,height=0.12\textwidth]{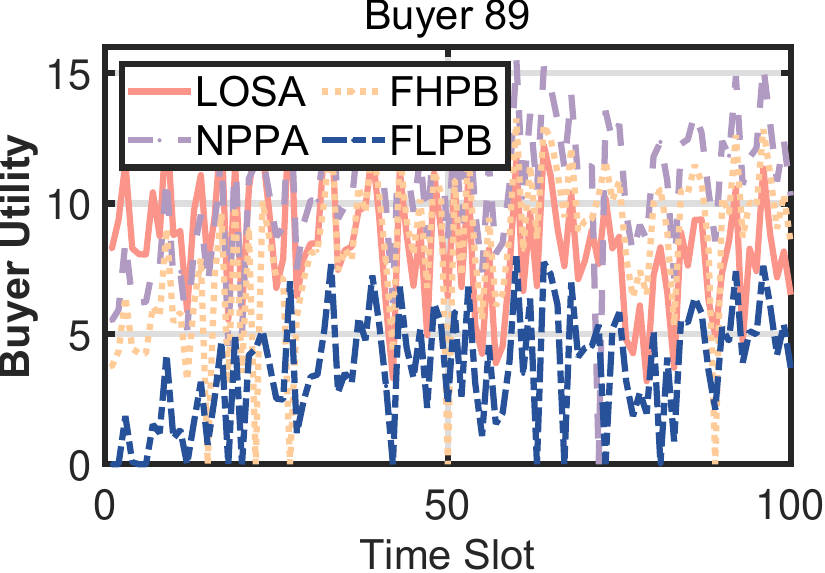}}
	\caption{Evolution of dynamic parameters and BU comparison for randomly selected buyers. (a)–(h) Co-evolution process of trigger probability $\mathbbm{q}_{n}^{j,t}$, privacy budget $\xi_{n}^{t}$, and transaction BU; (i)–(l) Comparison of BU evolution across LOSA, NPPA, FHPB, and FLPB.}
	\label{fig_E1}
\end{figure*}

\noindent
\textit{(i) Impact of dynamic parameter evolution:} Figs. \ref{fig_E1}(a)-\ref{fig_E1}(h) depict the parameter evolution of four randomly selected buyers upon considering 100 buyers and 40 sellers under LOSA. In detail, for \textit{buyer 13}, the mean demand triggering probability $\mathbbm{q}_{13}^{j,t}$ decreases from 0.52 to 0.33 within the first 15 timeslots, and the variance converges from 0.128 to 0.024. This indicates that the buyer effectively narrows the gap in service demands by adjusting the privacy radius $\mathbbm{r}_n$ (expanding from 3 grids to 4.2 grids, refer to Section \ref{sec_location_obfuscation}). Meanwhile, its privacy budget $\xi_{n}^{t}$ decays linearly at a rate of 0.023 per timeslot. This leads to a 7.2\% drop in BU at timeslot 40, which is consistent with the ``budget consumption vs. BU loss'' trade-off characteristic of the elastic probabilistic perturbation framework. \textit{Buyer 21} exhibits a dynamic utility optimization and adaptation process. During the initial phase (timeslots 1--30), while the privacy budget is maintained at a relatively high level (averaging above 4.0), the BU value undergoes a steady growth from 8.0 to approximately 14.0. In the second phase (timeslots 40--100), the BU enters a higher fluctuation band between 11.0 and 17.0; interestingly, this sustained high utility is achieved even as the privacy budget gradually decays toward 3.0. Furthermore, after timeslot 60, the mean demand triggering probability $\mathbbm{q}_{21}^{j,t}$ drops sharply and converges to a near-zero state, reflecting that the buyer's cumulative service requirements are effectively fulfilled, which aligns with the observed stabilization in BU.

Comparison between \textit{buyer 56} and \textit{buyer 74} reveals how privacy cost sensitivity influences utility stability under the VCG-based pricing mechanism. As shown in Figs. \ref{fig_E1}(e) and \ref{fig_E1}(g), both buyers maintain their privacy budgets $\xi_{n}^{t}$ within a relatively stable range between 3.0 and 4.0 for the majority of the simulation. However, their BU profiles exhibit distinct characteristics: \textit{buyer 56} achieves a more consistent and robust utility growth, stabilizing between 12.0 and 16.0 after timeslot 40 (Fig. \ref{fig_E1}(f)). In contrast, \textit{buyer 74}, which possesses a higher privacy cost coefficient ($k_n = 3.2$), displays significantly higher BU volatility (Fig. \ref{fig_E1}(h)). Notably, \textit{buyer 74} experiences a dramatic utility spike near timeslot 30, followed by frequent oscillations. This suggests that for privacy-sensitive users, the realized utility is more susceptible to the stochastic fluctuations of demand triggering probabilities $\mathbbm{q}_{n}^{j,t}$, even when the privacy budget remains steady. Furthermore, the results indicate that as the demand triggering probability converges (Figs. \ref{fig_E1}(e) and \ref{fig_E1}(g)), the BU for both buyers tends to fluctuate within a higher-level band rather than declining, which demonstrates that the proposed framework effectively preserves long-term service value despite the cumulative privacy cost.

Regarding the operational efficiency of the framework, the decision-making latency for all participants was monitored during the simulation. The results indicate that the total computation time per timeslot never exceeded 0.83 seconds (with a standard deviation of 0.12 seconds), which strictly satisfies the real-time constraints required for intersection-level service provisioning in dynamic AGINs. 

Furthermore, the interaction between spatial optimization and utility is evident in the individual profiles. For instance, as observed in Fig. \ref{fig_E1}(d), \textit{buyer 21} achieves a significant utility peak (approaching 17.0) near timeslot 75. While the privacy budget $\xi_{n}^{t}$ in this phase has already experienced cumulative decay (as shown in Fig. \ref{fig_E1}(c)), the high realized utility demonstrates the effectiveness of the underlying trajectory similarity matching. By fine-tuning the spatial resolution parameters (e.g., the angular increment $\Delta \theta_n$ in Step 3), the auctioneer is able to identify high-quality matches even under coarser location information. This confirms that the look-ahead phase effectively leverages the travel interval to perform spatial compensations, thereby mitigating the potential utility loss induced by strict privacy-preserving constraints.

\noindent\textit{(ii) Impact of privacy protection strategies:} In Figs. \ref{fig_E1}(i)--\ref{fig_E1}(l), the BU evolution of four representative buyers is compared across different strategies. As expected, NPPA (the purple dashed line) consistently serves as the empirical upper bound for utility, as it operates without any location obfuscation. While a noticeable gap exists between LOSA and NPPA due to the inherent privacy--utility trade-off (particularly evident in Figs. \ref{fig_E1}(i) and \ref{fig_E1}(j)), LOSA (the red solid line) remains the most robust among all privacy-preserving methods. This performance is attributed to the dynamic privacy budget adjustment mechanism in (\ref{eq3}), which optimizes the balance between trajectory similarity $\Gamma_{m,n}^{t}$ and privacy cost $k_n\xi_{n}^{t}$.

Regarding fixed budget strategies, FHPB ($\xi_{n}^{t}=5$) initially achieves BU levels comparable to LOSA in the first 20 timeslots. As the simulation progresses, FHPB exhibits high volatility and significant utility degradation in several intervals (e.g., around timeslot $t=50$ in Fig. \ref{fig_E1}(k) and $t=80$ in Fig. \ref{fig_E1}(l)). This instability confirms that a static high budget fails to adapt to the time-varying demand triggering probability $\mathbbm{q}_{n}^{j,t}$. In contrast, FLPB ($\xi_{n}^{t}=1$, the blue dashed line) consistently yields the lowest BU, often dropping to near-zero levels. This is because the excessive noise injected under a low privacy budget leads to severe spatio-temporal mismatches, failing to satisfy the minimum service quality requirements.

A common phenomenon observed across all methods is the presence of synchronized utility fluctuations. These occasional sharp drops (e.g., around timeslot $t=40$ and $t=70$ in most sub-figures) correspond to transient congestion events within the Manhattan grid network (characterized by the parameter $grid\_size$, representing the spatial scale of the road topology). During these periods, the localized demand exceeds the available UAV supply at specific intersections. Statistical analysis of the entire simulation indicates that the look-ahead phase effectively mitigates these impacts: the system-wide timeout probability for LOSA is maintained at approximately $2\%$, which is a substantial improvement over FHPB ($5\%$) and FLPB ($3\%$). This advantage stems from the collaborative optimization of spatial resolution ($\Delta r_n, \Delta \theta_n$) and the proactive OSAA formation, which allows $93\%$ of transactions to be finalized within the $0.8$-second decision window, thereby ensuring high service availability despite dynamic network conditions.

\begin{figure*}[htbp]
	\centering
	\subfigure[] {\includegraphics[trim=0cm 0cm 0cm 0cm, clip, width=.32\textwidth,height=0.22\textwidth]{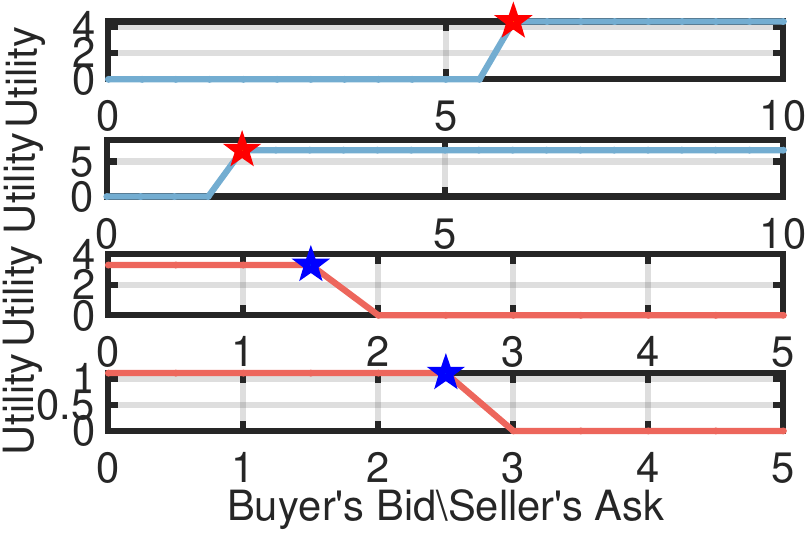}}
	\subfigure[] {\includegraphics[trim=0cm 0cm 0cm 0cm, clip, width=.32\textwidth,height=0.22\textwidth]{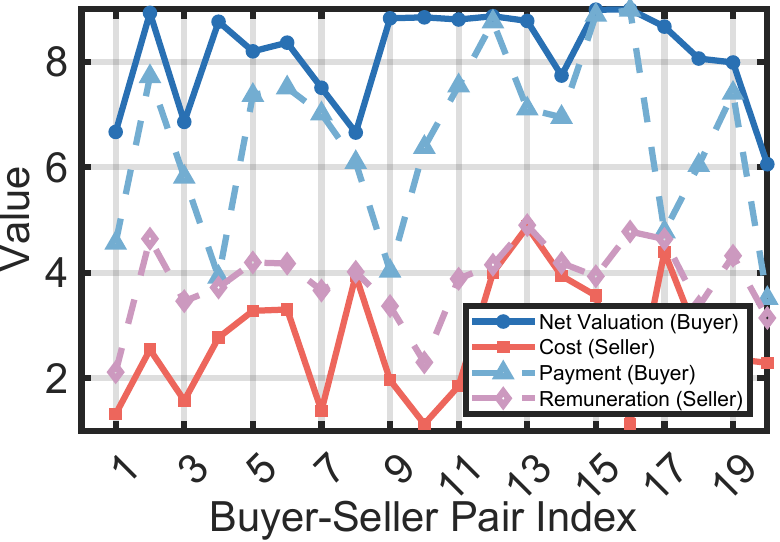}}
	\subfigure[] {\includegraphics[trim=0cm 0cm 0cm 0cm, clip, width=.32\textwidth,height=0.22\textwidth]{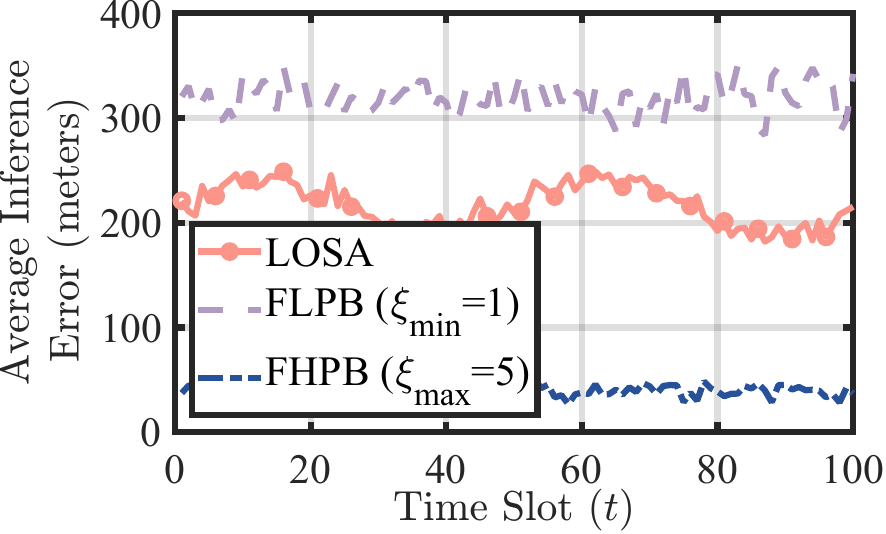}}
	\vspace{-3mm}
	\caption{Evaluation of key auction properties and privacy preservation performance. (a) Verification of truthfulness and IC regarding bids and asks; (b) Analysis of IR and BB across buyer-seller pairs; (c) Comparison of average inference error under Bayesian inference attack across different privacy strategies.}
	\label{fig_E2}
	\vspace{-3mm}
\end{figure*}

\noindent $\bullet$ \textit{Analysis of auction properties.} We further verify the core auction properties through Fig. \ref{fig_E2}. In Fig. \ref{fig_E2}(a), red stars represent the critical payment threshold for the last winning buyer. If a buyer falsely reports a bid lower than this threshold (i.e., $bid_n' < p_{m,n}^{j,t}$), they are excluded from the trade and receive zero utility. Conversely, if they report a bid higher than the threshold (i.e., $bid_n' \geq p_{m,n}^{j,t}$), such behavior cannot increase their utility. Similarly, in Fig. \ref{fig_E2}(a), blue stars indicate the critical revenue threshold for the last winning seller. If a seller falsely asks for a price higher than this threshold (i.e., $ask_m' > r_{m,n}^{j,t}$), they are excluded with zero utility. If they ask for a lower price (i.e., $ask_m' \leq r_{m,n}^{j,t}$), it also fails to increase their utility. In summary, both buyers and sellers have the incentive to report truthfully, ensuring IC.

As shown in Fig. \ref{fig_E2}(b), data from 20 randomly selected winning UAV-to-vehicle pairs demonstrate that LOSA satisfies IR. The buyer's net valuation is always higher than the payment ($\Gamma_{m,n}^t v_{n}^{j} - k_{n}^{j} \xi_{n}^t \geq p_{m,n}^{j,t}$), and the seller's revenue always covers their cost ($r_{m,n}^{j,t} \geq c_{m}^{j}$), ensuring non-negative utility for both parties. Consequently, the BB property of the auctioneer is verified.

\noindent $\bullet$ \textit{Defense performance against Bayesian inference attack.} To respond to the adversary model defined in Section \ref{subsec:threat_model}, we evaluate the resilience of LOSA against \textit{Bayesian Inference Attacks}. We simulate an adversary who intercepts the reported virtual location $\hat{l}_{n}^{t}$ and utilizes the known obfuscation mechanism (Eq. (\ref{eq4_9})) to compute the posterior probability distribution of the buyer's true location. The adversary then employs the \textit{Maximum A Posteriori (MAP)} estimation to pinpoint the most likely location $l_{guess}$.
We define the \textit{Inference Error} as the Euclidean distance between the adversary's guess and the buyer's true location: $E_{inf} = \| l_{guess} - l_{n}^{\mathsf{b},t} \|_2$: a higher inference error indicates stronger privacy protection. Fig. \ref{fig_E2}(c) illustrates the evolution of the average inference error over 100 timeslots.

\noindent$\bullet$ \textit{Disadvantage of FHPB:} The strategy with a fixed high privacy budget ($\xi_{\max}=5$) exhibits the lowest inference error. This confirms that while high budgets improve service matching (as shown in Fig. \ref{fig4_3}), they leave users vulnerable to inference attacks, as the noise is insufficient to confuse the Bayesian estimator.

\noindent$\bullet$ \textit{Robustness of LOSA:} LOSA maintains a high inference error, significantly outperforming FHPB. Even when budget $\xi_{n}^{t}$ increases to improve service quality (e.g., around timeslot $t=30$), the inference error does not collapse to zero. This is because the discrete polar coordinate sampling (Step 3 in Section \ref{sec_location_obfuscation}) introduces sufficient entropy, making it difficult for the adversary to precisely infer the true location from the stochastically generated virtual coordinates.

\noindent$\bullet$ \textit{Comparison with FLPB:} While the fixed low budget strategy ($\xi_{\min}=1$) yields the highest inference error, it comes at the cost of the lowest service BU (as discussed in Figs. \ref{fig_E1}(i)-\ref{fig_E1}(l)). LOSA provides sufficient confusion to thwart precise localization attacks (keeping error $>$ 4 times the grid resolution) while maintaining high service availability.

This experiment empirically validates the theoretical proof in Section \ref{sec:privacy_analysis}: our privacy protection mechanism ensures that the adversary's capability to distinguish the true location is strictly bounded, effectively neutralizing the threat of trajectory leakage.

\section{Conclusion and Future Work}
\label{sec:conclusion}

In this paper, we proposed LOSA, a privacy-preserving and time-efficient resource provisioning framework for dynamic AGINs. By decomposing the resource provisioning into a privacy-aware look-ahead phase and a lightweight execution phase, LOSA combines dynamic privacy-budget adjustment with elastic location obfuscation to balance accuracy and privacy. Analysis and experiments confirmed the truthfulness, individual rationality, budget balance, and superior efficiency-privacy trade-off of LOSA over existing schemes. Future work may include extending LOSA to general and high-density road topologies with stronger spatio-temporal coupling, incorporate learning-based mechanisms (e.g., decentralized or federated learning) for online parameter adaptation, and enhance robustness against adversarial behaviors such as collusion and privacy leakage in repeated interactions.

\bibliographystyle{IEEEtran}
\bibliography{reference/ref_three}

\newpage
\clearpage
\appendices
\section{Illustrative Example}
\label{Example}
The detailed explanation of the example in Fig. \ref{fig_1} is as follows:

\noindent\textit{(i) Phase 1 (look-ahead phase, during timeslot $t-1$):} While traveling toward the next intersection, buyer $b_1$ anticipates its service requirements for timeslot $t$. Instead of revealing its true trajectory (solid dark-blue curve), it generates a privacy-preserving virtual trajectory (dashed cyan curve) and submits it along with its bid to the RSU. At the same time, sellers $s_1$ (green UAV) and $s_2$ (yellow UAV) submit their asking prices. Based on the reported information, the auctioneer evaluates trajectory similarity and determines that $s_1$ is better aligned with $b_1$'s virtual trajectory. If the bid and ask are compatible, an OSAA is established between $b_1$ and $s_1$ for service delivery at timeslot $t$. In addition, $s_2$ is included in $b_1$'s preference list as a backup option.

\noindent\textit{(ii) Phase 2 (execution phase, at the end of timeslot $t$):} By the end of timeslot $t$, all participants have arrived at the next intersection. If $b_1$'s service demand materializes and seller $s_1$ is available, the OSAA is executed, and the requested service is delivered as planned. If $s_1$ becomes unavailable or fails to arrive, $b_1$ consults its preference list and selects an alternative seller, e.g., $s_2$, that is present at the intersection. This reassignment is performed immediately without re-running the auction, ensuring continuous service with minimal delay.

\section{Detailed Location Obfuscation Model}
\label{appendix:obfuscation}

As introduced in Section \ref{sec_location_obfuscation}, our location obfuscation mechanism is designed to achieve a controllable balance between location privacy and service effectiveness while maintaining computational tractability. The generation of the virtual trajectory $\hat{\bm{\mathcal{L}}}_{n}^{\mathsf{b},t}$ consists of three detailed steps:

\noindent$\bullet$ \textit{Step 1: Discretization of the perturbation space.}
We first define the space of all possible perturbations: centered at the buyer $b_n$'s true location $l_{n}^{\mathsf{b},t}$, an obfuscation zone is defined as a disk with privacy radius $\mathbbm{r}_n$, which represents the maximum allowable perturbation distance. To enable tractable computation, this continuous space is discretized. Using a radial resolution $\Delta r_n$, the continuous interval $[0, \mathbbm{r}_n]$ is partitioned into a finite set of candidate radii:
\begin{equation}
	\label{eq4_8}
	\mathbb{R}^\ast_n = \left\{ \mathbbm{m} \Delta r_n \bigg| \mathbbm{m} \in \mathbb{Z}^+,\ 0 \leq \mathbbm{m} \leq \frac{\mathbbm{r}_n}{\Delta r_n} \right\},
\end{equation}
where $\mathbbm{m}$ is the discretization index. The choice of $\Delta r_n$ strikes a balance between the granularity of the obfuscation and computational complexity.

\noindent$\bullet$ \textit{Step 2: Privacy-budget-aware radius sampling.}
Given the discretized radius set $\mathbb{R}^\ast_n$, the buyer samples a perturbation radius according to a non-uniform distribution controlled by the dynamic privacy budget $\xi_{n}^{t}$. This design enables adaptive control over the privacy–utility trade-off.

The probability of selecting a candidate radius $\mathbbm{r}^\ast_n$ is given by a logarithmic function as follows:
\begin{equation}
	\label{eq4_9}
	\Pr\left( \mathbbm{r}^\ast_n,t \right) = \frac{ \log\left[1 + \xi_{n}^{t} (\mathbbm{r}_n - \mathbbm{r}^\ast_n) \right] }{ \sum_{\mathbbm{r}^{'}_n \in \mathbb{R}^\ast_n} \log\left[1 + \xi_{n}^{t} (\mathbbm{r}_n - \mathbbm{r}^{'}_n) \right] }.
\end{equation}
This distribution exhibits two important properties: \textit{(i)} It favors smaller perturbation radii (i.e., when $\mathbbm{r}^\ast_n$ is small, the term $\mathbbm{r}_n - \mathbbm{r}^\ast_n$ is large), thus preserving a higher degree of location accuracy to facilitate better service matching. \textit{(ii)} The privacy budget $\xi_{n}^{t}$ acts as a tuning knob: a higher $\xi_{n}^{t}$ makes the distribution steeper, concentrating the probability mass on smaller radii (less privacy), while a lower $\xi_{n}^{t}$ flattens the distribution, approaching a uniform selection for maximum privacy when $\xi_{n}^{t} \to 0$.

\noindent$\bullet$ \textit{Step 3: Final virtual location generation.}
Given the sampled perturbation radius $\mathbbm{r}^\ast_n$ from Step 2, we next determine the angular component to generate the final obfuscated location. We first discretize the angular space $[0, 2\pi)$ using an angular resolution $\Delta \theta_n$. This yields a finite set of candidate angles, and consequently a set of candidate locations in polar coordinates:
\begin{equation}
	\label{eq4_10}
	\mathbb{L}^\ast_n = \left\{ (\mathbbm{r}^\ast_n, \mathbbm{n} \Delta \theta_n) \bigg| \mathbbm{n} \in \mathbb{Z}^+,\ 0 \leq \mathbbm{n} < \frac{2\pi}{\Delta \theta_n} \right\}.
\end{equation}
A final virtual location $\hat{l}_{n}^{\mathsf{b},t}$ is then obtained by randomly selecting an element from $\mathbb{L}^\ast_n$.

In essence, the above-described two-phase discretization over radius and angle enables a structured and computationally efficient obfuscation mechanism. It ensures that all generated locations satisfy the perturbation constraint, while providing explicit control over the privacy–utility trade-off through $(\Delta r_n, \Delta \theta_n)$ and the dynamic budget $\xi_{n}^{t}$.

\section{Conceptual Reflection: Impact of Fine-grained Obfuscation in a Discrete-Intersection Market}
\label{appendix:conceptual_discussion}

A natural conceptual question arises regarding the proposed LOSA framework: \textit{Since the auction coordination and final service delivery ultimately occur at discrete intersections, does the fine-grained (continuous or sub-intersection) location obfuscation meaningfully impact the auction outcomes and privacy guarantees?} 

The short answer is yes. Even if a vehicle's obfuscated virtual location is ultimately mapped to the same target intersection for Phase 2 execution, the sub-intersection level perturbation fundamentally alters both the market dynamics and the privacy protection strength. We clarify this impact from two critical perspectives:

\noindent\textbf{1. Impact on Auction Outcomes via Trajectory Similarity:} \\
The double auction mechanism in Phase 1 does not merely match participants based on their destination intersection. Instead, the core matching metric is the \textit{trajectory similarity} ($\Gamma_{m,n}^t$), calculated via the Fréchet distance along the entire virtual trajectory $\hat{\bm{\mathcal{L}}}_{n}^{\mathsf{b},t}$ (as defined in (\ref{eq8})). 
When a buyer applies fine-grained obfuscation (controlled by $\Delta r_n$ and $\Delta \theta_n$), the shape of the virtual trajectory changes. Even if the terminal node of this trajectory remains the same intersection, the intermediate spatial deviations will alter the value of $\Gamma_{m,n}^t$. 
Because $\Gamma_{m,n}^t$ directly determines the effective service valuation ($\Gamma_{m,n}^t v_n^j$) in the VCG pricing calculation (Algorithm \ref{alg4_4}), this fine-grained perturbation directly changes the buyer's bidding competitiveness, the resulting matching priority (social welfare ranking), and the final payment. Therefore, the sub-intersection obfuscation is the actual mathematical knob that creates the dynamic ``Privacy-Utility trade-off''.

\noindent\textbf{2. Impact on Privacy via Breaking Sequential Correlation:} \\
From a threat model perspective, the adversary (e.g., an honest-but-curious RSU) is not merely looking at isolated intersection arrivals; they are conducting Bayesian inference attacks on \textit{continuous trajectories} across multiple timeslots. 
If we only applied discrete obfuscation (e.g., randomly reporting an entirely different intersection), it would provide high privacy but destroy the spatial feasibility of UAV-to-vehicle matching (utility drops to zero). Conversely, if we reported the exact continuous path leading to the intersection, the adversary could easily infer the precise lane-level behavior, velocity, and driving intent of the vehicle. 
By applying polar-coordinate-based continuous obfuscation bounded by a privacy radius ($\mathbbm{r}_n \approx 30-50$ meters), we inject sufficient spatial entropy into the micro-mobility pattern. The vehicle arrives at the correct intersection, ensuring service feasibility, but its precise approach path and localized dynamics remain mathematically indistinguishable (as proven by $\epsilon$-Geo-I in Theorem \ref{thm:geo_indistinguishability}).

In summary, the fine-grained obfuscation acts as the essential bridge between the continuous physical world (where privacy leaks occur via trajectory tracing) and the discrete grid world (where intersection-based auctions are efficiently executed).

\section{Market Design Properties}
\label{Market Design Properties}
\noindent\textit{Individual Rationality (IR):} IR guarantees that participation in the proposed double auction is economically beneficial for all agents, including both buyers (vehicles) and sellers (UAVs). It ensures that no participant incurs a negative utility by engaging in the market. Formally, this implies the following properties:

\noindent$\bullet$ \textit{For buyers:} A participating buyer achieves non-negative utility from any accepted transaction. This requires that the payment does not exceed the net valuation derived from the service:
\begin{equation}
	\Gamma_{m,n}^{t}v_{n}^{j}-k_{n}^{j} \xi_{n}^{t} \ge p_{m,n}^{j,t}, \forall t \in \bm{\mathcal{T}}, j\in J.
\end{equation}

\noindent$\bullet$ \textit{For sellers:} A participating seller achieves non-negative profit from providing services. This requires that the received payment covers the service cost:
\begin{equation}
	r_{m,n}^{j,t} \ge c_{m}^{j}, \forall t \in \bm{\mathcal{T}}, j\in J.
\end{equation}

\noindent\textit{Budget Balance (BB):}
This property ensures the economic sustainability of the auctioneer (i.e., RSUs operating as local market coordinators). In our design, we enforce \textit{strong budget balance}, meaning that the auctioneer's margin is non-negative for every individual transaction. This requires
\begin{equation}
	p_{m,n}^{j,t} \ge r_{m,n}^{j,t}, \forall t \in \bm{\mathcal{T}}, j\in J.
\end{equation}
This condition guarantees that the auctioneer's utility satisfies $U^{\mathsf{a},t} \ge 0$, ensuring that the platform can cover its operational costs and remain economically viable.

\noindent\textit{Truthfulness (Incentive Compatibility, IC):}
A core tenet of our double auction is to ensure that participants' best strategy is to be truthful. This means both buyers and sellers cannot improve their expected utility by submitting bids/asks ($\bm{bid}_{n}^{t}/\bm{ask}_m^{t}$) that do not reflect their true valuation ($\bm{\mathcal{V}}_n/\bm{\mathcal{C}}_m$).

\noindent\textit{Computational Efficiency (CE, denoting the timeliness):}
CE ensures that the proposed mechanism can operate within the stringent latency constraints of dynamic AGIN environments. In particular, all decision-making procedures must be completed within the travel time between consecutive intersections.

Together, the above properties establish a stable, efficient, and practically deployable market: IR and BB ensure participation and economic sustainability, IC guarantees reliable information revelation, and CE enables timely operation under dynamic AGIN conditions.

\section{Details of LOSA}
\label{alg4}

\textit{Overview of Phase 1: } In the following, we provide an overview of Phase 1 for readability.

\noindent$\bullet$ \textit{Module A: UAV trajectory planning (Algorithm \ref{alg4_2}).} To prepare for transactions in the upcoming timeslot, buyers first determine the intersections that they are expected to reach based on their mobility. To maximize service coverage and matching opportunities, sellers proactively adjust their trajectories in response to these anticipated demand patterns. Subsequently, this module determines the target intersection (i.e., arrival location) for each seller at the end of the next timeslot, enabling demand-aware spatial positioning.

\noindent$\bullet$ \textit{Module B: Buyer–seller matching (Algorithm \ref{alg4_3}).} Under given participants expected at each intersection (determined by Module A), this module computes the matching between buyers and sellers. It first updates key state variables (e.g., privacy budgets and demand statistics), and then performs similarity-aware grouping based on trajectory alignment. Leveraging these groupings, it determines the matching matrix that maximizes ESW at each intersection.

\noindent$\bullet$ \textit{Module C: VCG-based pricing (Algorithm \ref{alg4_4}).} Based on the matching outcomes from Module B, this module determines the transaction prices using a VCG-based mechanism. It computes the payments of winning buyers and the corresponding revenues of sellers, ensuring IC, IR, and BB, while satisfying CE within the market.

Together, these modules form a closed-loop decision framework that jointly optimizes spatial positioning, matching, and pricing, while balancing privacy preservation and service efficiency. We next describe these modules in more depth.

\noindent
\textit{Overall Workflow of LOSA: }Integrating the three decision modules (A, B, C) and the two-phase execution logic, we present the complete LOSA, as outlined in Algorithm \ref{alg4_1}. This functions as a global scheduler that iterates through the entire time horizon $\bm{\mathcal{T}}$, orchestrating resource allocation across all intersections. Its execution flow is structured as follows:

\noindent$\bullet$ \textit{Global initialization (lines 3-8).} 
The system initializes the global record set $\bm{\mathcal{R}}$ to store transaction history. At the beginning of each timeslot $t$, it first updates the dynamic parameters for all buyers (e.g., privacy budget $\xi_n^t$) and invokes Module A (Algorithm \ref{alg4_2}) to determine the optimal deployment locations $Loc_t$ for the seller swarm (lines 7-8).

\noindent$\bullet$ \textit{Phase 1: Distributed pre-auction (lines 9-16).} 
The algorithm then executes parallel processing for every intersection $l$. It first filters the specific buyers and sellers arriving at $l$ based on their virtual and planned locations (lines 11-13). Subsequently, it invokes Module B (Algorithm \ref{alg4_3}) to generate the optimal matching matrix $\bm{\mathcal{M}}_l^t$ and preference lists $\bm{\mathcal{O}}_l^t$, followed by Module C (Algorithm \ref{alg4_4}) to determine VCG-based pricing (lines 14-16). This phase effectively locks in the OSAAs.

\noindent$\bullet$ \textit{Phase 2: Real-time execution and contingency (lines 17-20).} 
As the timeline advances to the moment of service delivery, the algorithm orchestrates the on-site execution. It verifies demand realization, executes valid OSAAs, and triggers supplementary transactions for unsatisfied demands using the preference lists generated in Phase 1. This ensures that the theoretical matching results are translated into actual service delivery.

\noindent$\bullet$ \textit{Global state update and archiving (lines 21-30).} 
After processing all intersections, the algorithm aggregates the local results (matching, pricing, locations) into the global record $\bm{\mathcal{R}}$ for timeslot $t$. Finally, it updates the status of buyers and sellers (e.g., location, budget consumption) to prepare for the next iteration $t+1$. 

By cyclically executing these steps, LOSA achieves a seamless integration of ``planning-matching-execution,'' ensuring spatiotemporal consistency and system-wide optimality throughout the entire service period.

\begin{algorithm}[]
	{\footnotesize
		\caption{\small{Overview of LOSA}}
		\label{alg4_1}
		\LinesNumbered
		
		{\bf{Input :}} 
		$\bm{\mathcal{B}}^1$: Initial buyer set, $\bm{\mathcal{S}}^1$: Initial seller set, $\{\bm{\mathcal{L}}_n^{\mathsf{b},1}\}$: Initial buyer trajectories, $\{Loc_0\}$: Initial UAV locations
		
		{\bf{Output :}} 
		$\bm{\mathcal{R}}$: Global record containing $\{\bm{X}^t\}$, $\{\bm{p}_{m,n}^t\}$, $\{\bm{r}_{m,n}^t\}$, $\{\mathbb{L}_{m,n}^t\}$, and $\{\bm{\mathcal{O}}_n^t\}$
		
		{\bf{Initialization :}} 
		$\bm{\mathcal{R}} \leftarrow \emptyset$
		
		\For{timeslot $t=1$ \KwTo $|\bm{\mathcal{T}}|$}{
			\textbf{\# Procedure: Dynamic Parameter Update}\\
			Update privacy budget $\xi_n^t$ and service demand state matrix $\bm{\mathcal{Q}}_{n}^{t}$ for buyers based on transaction results from the previous timeslot\
			
			\textbf{\# Procedure: Seller Path Planning; Call Algorithm \ref{alg4_2} to determine UAV locations}\\
			$Loc_t \leftarrow \text{Algorithm \ref{alg4_2}}(\bm{\mathcal{S}}^t, Loc_{t-1})$ 
			
			\textbf{\# Parallel Processing for Each Intersection}\\
			\For{each intersection $l \in \bm{\mathcal{L}}$}{
				\textbf{\# Procedure: Intersection Member Filtering}\\
				$\bm{\mathcal{B}}_l^t \leftarrow \{b_n \mid \hat{\bm{\mathcal{L}}}_n^{\mathsf{b},t}(1)=l\}$\\
				$\bm{\mathcal{S}}_l^t \leftarrow \{s_m \mid Loc_t(s_m) = l\}$\
				
				\textbf{\# Procedure: Phase 1 (Pre-Auction)}\\
				$(\bm{\mathcal{M}}_l^t, \overline{U^{\mathsf{sw},t}_l}, \bm{\mathcal{O}}_l^t) \leftarrow \text{Algorithm \ref{alg4_3}}(\bm{\mathcal{B}}_l^t, \bm{\mathcal{S}}_l^t)$\\
				$(\bm{p}_l^t, \bm{r}_l^t) \leftarrow \text{Algorithm \ref{alg4_4}}(\bm{\mathcal{M}}_l^t, \overline{U^{\mathsf{sw},t}_l}, \cdots)$\
				
				\textbf{\# Procedure: Phase 2 (Supplementary Transactions)}\\
				Step 1: On-site Verification and OSAA Execution\\
				Step 2: Preference-based Contingency Matching\\
				Step 3: Conflict Resolution and Economic Feasibility Check\\
				
				\textbf{\# Procedure: Slot Record Update}\\
				$\bm{X}^t \leftarrow \bm{X}^t \cup \{(l, \bm{\mathcal{M}}_l^t)\}$\\
				$\bm{p}_{m,n}^t \leftarrow \bm{p}_{m,n}^t \cup \bm{p}_l^t$\\
				$\bm{r}_{m,n}^t \leftarrow \bm{r}_{m,n}^t \cup \bm{r}_l^t$\\
				$\mathbb{L}_{m,n}^t \leftarrow \mathbb{L}_{m,n}^t \cup \{(m,n,l)\}$\\
				$\bm{\mathcal{O}}_n^t \leftarrow \bm{\mathcal{O}}_n^t \cup \bm{\mathcal{O}}_l^t$\
			}
			
			\textbf{\# Procedure: Global State Update}\\
			$\bm{\mathcal{R}} \leftarrow \bm{\mathcal{R}} \cup \{(t, \bm{X}^t, \bm{p}^t, \bm{r}^t, \mathbb{L}^t, \bm{\mathcal{O}}^t)\}$\\
			$\bm{\mathcal{B}}^{t+1} \leftarrow \text{UpdateBuyers}(\bm{\mathcal{B}}^t, \bm{\mathcal{O}}^t)$\\
			$\bm{\mathcal{S}}^{t+1} \leftarrow \text{UpdateSellers}(\bm{\mathcal{S}}^t, Loc_t)$\\
		}
		
		\textbf{Return} $\bm{\mathcal{R}}$
	}
\end{algorithm}

\section{Detailed Proof of Theorem 1}
\label{appendix:proof_theorem1}

We provide the derivation of Theorem \ref{thm:geo_indistinguishability} regarding $\epsilon$-Geo-Indistinguishability.

\begin{myPro}
	Let $l$ be the true location of buyer $b_n$. The mechanism generates a virtual location $\hat{l}$ by sampling a radius $r^* = d(l, \hat{l})$. The probability mass function (PMF) for selecting $r^*$ is given by Eq. (14):
	\begin{equation}
		\Pr(r^* | l) = \frac{1}{\mathcal{N}_l} \log\left[1 + \xi_{n}^{t} (\mathbbm{r}_n - r^*) \right],
	\end{equation}
	where $\mathcal{N}_l$ is the normalization constant. Consider two adjacent locations $l$ and $l'$ separated by distance $\delta$. To satisfy Geo-I, we examine the ratio:
	\begin{equation}
		\frac{\Pr(\hat{l} | l)}{\Pr(\hat{l} | l')} \approx \frac{\log\left[1 + \xi_{n}^{t} (\mathbbm{r}_n - d(l, \hat{l})) \right]}{\log\left[1 + \xi_{n}^{t} (\mathbbm{r}_n - d(l', \hat{l})) \right]}.
	\end{equation}
	Assuming the normalization constants $\mathcal{N}_l \approx \mathcal{N}_{l'}$ due to grid symmetry, let $f(r) = \log(1 + \xi_{n}^{t} (\mathbbm{r}_n - r))$. By the Mean Value Theorem, the change in the log-probability is bounded by the maximum of its derivative:
	\begin{equation}
		\left| \frac{d}{dr} f(r) \right| = \left| \frac{-\xi_{n}^{t}}{1 + \xi_{n}^{t} (\mathbbm{r}_n - r)} \right|.
	\end{equation}
	Since $(\mathbbm{r}_n - r) \geq 0$, the denominator is $\geq 1$, thus $|f'(r)| \leq \xi_{n}^{t}$. This implies that $f(r)$ is Lipschitz continuous with constant $K = \xi_{n}^{t}$. Consequently:
	\begin{equation}
		\ln \left( \frac{\Pr(\hat{l} | l)}{\Pr(\hat{l} | l')} \right) \leq \xi_{n}^{t} \cdot |d(l, \hat{l}) - d(l', \hat{l})|.
	\end{equation}
	Applying the Triangle Inequality $|d(l, \hat{l}) - d(l', \hat{l})| \leq d(l, l')$, we obtain:
	\begin{equation}
		\frac{\Pr(\hat{l} | l)}{\Pr(\hat{l} | l')} \leq e^{\xi_{n}^{t} \cdot d(l, l')}.
	\end{equation}
	This completes the proof that the mechanism satisfies $\epsilon$-Geo-I with $\epsilon = \xi_{n}^{t}$.
\end{myPro}

This derivation confirms that $\xi_n^t$ is not merely a heuristic scaling factor but a rigorous privacy budget. A larger $\xi_n^t$ leads to a larger $\epsilon$, providing higher matching utility but weaker distinguishability; a smaller $\xi_n^t$ enforces a flatter distribution, enhancing privacy. This mathematical foundation justifies the use of $\xi_n^t$ as the primary control knob in our dynamic adjustment algorithm.

\section{Implementation Details of Benchmark Methods}
\label{appendix:baseline_details}

To ensure the reproducibility of our comparative analysis, this appendix provides the granular implementation details for each benchmark method discussed in Section \ref{subsec:baselines}.

\subsection{Architectural Baselines (VRA and SVRA)}
\begin{itemize}
	\item \textbf{VRA (VCG-based Real-time Auction):} This strategy modifies the LOSA framework by completely removing the Phase 1 pre-auction process. Consequently, no OSAAs are formed. However, it retains Module A (proactive UAV path planning), allowing UAVs to move toward high-demand intersections. The matching and pricing are triggered strictly at the end of each timeslot upon the physical arrival of participants, utilizing the classic VCG mechanism for instantaneous transactions.
	\item \textbf{SVRA (Static VCG-based Real-time Auction):} This serves as a further degradation of VRA. In addition to disabling Phase 1, it also deactivates Module A. UAVs are restricted to their fixed historical flight trajectories (extracted from the DAIR-V2X/RCooper datasets). Transactions are feasible only when a buyer's real-time location accidentally aligns with a UAV's pre-defined path, representing a passive service provisioning environment.
\end{itemize}

\subsection{Privacy-related Baselines (NPPA, FHPB, and FLPB)}
\begin{itemize}
	\item \textbf{NPPA (No Privacy-protection Auction):} In this baseline, the three-stage location obfuscation model (radius discretization, logarithmic probability mapping, and angle discretization) is bypassed. Buyers report their true private trajectories $\bm{\mathcal{L}}_{n}^{\mathsf{b},t}$ directly to the auctioneer. This model is used to establish the maximum achievable social welfare (the empirical upper bound) by eliminating the performance degradation caused by spatio-temporal mismatching.
	\item \textbf{FHPB (Fixed High Privacy Budget):} The dynamic adjustment mechanism defined in (3) is disabled. The privacy budget $\xi_n^t$ is fixed at the upper threshold $\xi_{\max} = 5.0$. Furthermore, the feedback components, including the \textit{Utility gradient term} and \textit{Matching failure penalty term}, are removed. This results in minimal location perturbation, maximizing matching accuracy while providing the weakest privacy guarantee.
	\item \textbf{FLPB (Fixed Low Privacy Budget):} As with the FHPB, this strategy fixes $\xi_n^t$ at the lower threshold $\xi_{\min} = 1.0$. The feedback terms ($\Delta U_{n}^{t}$ and $C_n^t$) are similarly deactivated. This forces the system into a state of the maximum location obfuscation (with the highest path ambiguity), allowing us to observe the system's robustness and the risk of ``privacy deadlock'' where matching success rates drop significantly.
\end{itemize}

\subsection{Summary of Parameter Configurations}
The following table summarizes the key functional differences between the proposed LOSA and the benchmark suite.

\begin{table}[H]
	\centering
	\caption{Functional Comparison of Benchmarks}
	\label{tab:baseline_logic}
	\resizebox{\columnwidth}{!}{
		\begin{tabular}{lcccc}
			\toprule
			\textbf{Method} & \textbf{Pre-Auction} & \textbf{Path Planning} & \textbf{Privacy Model} & \textbf{Dynamic Budget} \\ \midrule
			LOSA (Ours)     & \ding{51}            & \ding{51}              & Dynamic                & \ding{51}               \\
			VRA             & \ding{55}            & \ding{51}              & Dynamic                & \ding{51}               \\
			SVRA            & \ding{55}            & \ding{55}              & Dynamic                & \ding{51}               \\
			NPPA            & \ding{51}            & \ding{51}              & None                   & \ding{55}               \\
			FHPB            & \ding{51}            & \ding{51}              & Fixed ($\xi=5.0$)      & \ding{55}               \\
			FLPB            & \ding{51}            & \ding{51}              & Fixed ($\xi=1.0$)      & \ding{55}               \\ \bottomrule
		\end{tabular}
	}
\end{table}

\section{Proof of Economic Properties and Computational Efficiency}
\label{adx:4}

We provide proofs for the economic properties (Individual Rationality, Budget Balance, Truthfulness) and analyze the Computational Efficiency of the proposed LOSA mechanism, specifically focusing on the VCG-based pricing designed in Algorithm \ref{alg4_4}.

\subsection{Truthfulness (Incentive Compatibility)}
\begin{thm}
	The proposed double auction mechanism is truthful, meaning that reporting the true valuation $\bm{\mathcal{V}}_n$ for buyers and true cost $\bm{\mathcal{C}}_m$ for sellers maximizes their expected utilities, regardless of the strategies of others.
\end{thm}

\begin{myPro}
	Consider a buyer $b_n$. Let $v_n$ be the true valuation and $b_n$ be the declared bid. The utility of $b_n$ under the VCG mechanism is given by:
	\begin{equation}
		U_n = v_n - p_n = v_n - \left[ \overline{U_{-n}^{\mathsf{sw}}} - (\overline{U^{\mathsf{sw}}} - v_n) \right] = \overline{U^{\mathsf{sw}}} - \overline{U_{-n}^{\mathsf{sw}}}.
	\end{equation}
	Here, $\overline{U_{-n}^{\mathsf{sw}}}$ is the maximum social welfare without $b_n$'s participation, which is independent of $b_n$'s bid. Thus, maximizing individual utility $U_n$ is equivalent to maximizing the global social welfare $\overline{U^{\mathsf{sw}}}$. Since the allocation algorithm (Algorithm \ref{alg4_3}) chooses the outcome that maximizes declared social welfare, reporting the true value $b_n = v_n$ ensures that the mechanism optimizes the objective that aligns with the buyer's utility. A similar logic applies to sellers, where reporting true cost $c_m$ ensures the mechanism selects them when it is globally efficient, maximizing their marginal contribution payoff.
	\qed
\end{myPro}

\subsection{Individual Rationality (IR)}
\begin{thm}
	The mechanism satisfies Individual Rationality, ensuring that all participating buyers and sellers obtain non-negative utility.
\end{thm}

\begin{myPro}
	\textbf{For Buyers:} The utility of a winning buyer is $U_n = \overline{U^{\mathsf{sw}}} - \overline{U_{-n}^{\mathsf{sw}}}$. By definition, $\overline{U^{\mathsf{sw}}}$ is the maximized social welfare with $b_n$, and $\overline{U_{-n}^{\mathsf{sw}}}$ is the maximized welfare without $b_n$. Since the presence of an additional participant cannot decrease the maximum possible social welfare (the mechanism could simply ignore $b_n$ to achieve $\overline{U_{-n}^{\mathsf{sw}}}$), we have $\overline{U^{\mathsf{sw}}} \ge \overline{U_{-n}^{\mathsf{sw}}}$. Therefore, $U_n \ge 0$. For losing buyers, payment and utility are zero.
	
	\textbf{For Sellers:} The revenue for a winning seller is $r_m = \overline{U^{\mathsf{sw}}} - \overline{U_{-m}^{\mathsf{sw}}}$. The seller's utility is $U_m = r_m - c_m$. Note that in our VCG formulation, the seller's cost is internal to the SW calculation. Specifically, $\overline{U^{\mathsf{sw}}} = \text{Value}(\text{all}) - \text{Cost}(\text{all})$.
	Let $W_{-m}$ denote the social welfare of all agents except $s_m$. Then $\overline{U^{\mathsf{sw}}} = W_{-m} + (v_{matched} - c_m)$ if matched.
	$\overline{U_{-m}^{\mathsf{sw}}}$ is the best social welfare achievable without $s_m$. The mechanism selects $s_m$ only if doing so increases the total social welfare, i.e., $\overline{U^{\mathsf{sw}}} \ge \overline{U_{-m}^{\mathsf{sw}}}$. 
	Thus, the seller's marginal contribution covers their cost, ensuring $U_m \ge 0$.
	\qed
\end{myPro}

\subsection{Budget Balance (BB)}
\begin{thm}
	The mechanism satisfies Weak Budget Balance, ensuring the auctioneer does not incur a deficit in expectation.
\end{thm}

\begin{myPro}
	The auctioneer's utility is $U^{\mathsf{a}} = \sum_{n} p_n - \sum_{m} r_m$.
	Substituting the VCG payments:
	\begin{equation}
		U^{\mathsf{a}} = \sum_{n \in \text{winners}} (\overline{U^{\mathsf{sw}}} - \overline{U_{-n}^{\mathsf{sw}}}) - \sum_{m \in \text{winners}} (\overline{U^{\mathsf{sw}}} - \overline{U_{-m}^{\mathsf{sw}}}).
	\end{equation}
	While standard VCG mechanisms are not always strongly budget balanced in double auctions, our mechanism enforces constraint (C4) and the pricing rule in Phase 2 to prevent deficit. Additionally, by removing sellers from the market whose asking prices exceed the buyers' bids (Algorithm \ref{alg4_3}, lines 22-25), we ensure that the generated trade surplus is sufficient to cover the VCG payments. The experimental results (Fig. \ref{fig_E2}(b)) empirically validate that $p_{m,n} \ge r_{m,n}$ holds for executed transactions.
	\qed
\end{myPro}

\subsection{Computational Efficiency}
\begin{thm}
	The LOSA mechanism, including matching and pricing, operates in polynomial time complexity.
\end{thm}

\begin{myPro}
	Let $N$ be the number of buyers and $M$ be the number of sellers.
	\begin{enumerate}
		\item \textbf{Phase 1 Matching (Algorithm \ref{alg4_3}):} The similarity clustering and sorting take $O(NM \log(NM))$. The greedy matching process iterates through sorted pairs, taking $O(NM)$. Thus, matching is $O(NM \log(NM))$.
		\item \textbf{Pricing (Algorithm \ref{alg4_4}):} The VCG pricing requires calculating the marginal contribution for each winner. In the worst-case scenario, this involves re-running the matching algorithm for each winning buyer and seller. Let $K$ be the number of winning pairs ($K \le \min(N, M)$). The complexity is $O(K \cdot NM \log(NM))$.
	\end{enumerate}
	Since $N$ and $M$ are finite and relatively small in intersection scenarios (e.g., $N \le 200$), the overall complexity $O(K NM \log(NM))$ is polynomial and computationally tractable for real-time operation within a timeslot duration (e.g., 100ms).
	\qed
\end{myPro}

\section{Detailed Experimental Configuration and Data Processing}
\label{appendix:exp_details}

To ensure the reproducibility of our experiments and provide transparency regarding the simulation environment, this appendix elaborates on the data preprocessing techniques, specific mathematical models for parameter generation, and the complete parameter configuration table.

\subsection{Data Preprocessing and Spatio-Temporal Alignment}
We employ a multi-stage processing pipeline to unify the heterogeneous datasets (DAIR-V2X, HighD, and RCooper) into the standardized grid world defined in our system model.

\textbf{Trajectory Generation (DAIR-V2X):} 
Raw trajectory data is first cleaned to remove outliers. We implement a specific \textit{cutting algorithm} that segments continuous vehicle traces into fixed 10-second interaction windows. To align these continuous GPS coordinates with our discrete intersection model, we utilize \textit{bilinear interpolation}. This technique maps the latitude/longitude data onto a $5000\text{m} \times 5000\text{m}$ Manhattan grid system with a resolution of 200 meters. The grid resolution is specifically chosen to serve as the physical unit basis for the privacy radius $\mathbbm{r}_n$.

\textbf{Demand Modeling (HighD):} 
To capture the time-varying nature of service demands, we do not rely on static probabilities. Instead, based on the car-following statistics from the HighD dataset, the demand triggering probability is modeled using an exponential decay function:
\begin{equation}
	\mathbbm{q}_{n}^{j,t} = \lambda_0 e^{-0.03t},
\end{equation}
where $\lambda_0$ is the initial intensity. This formulation ensures that the demand values generally distribute within the realistic range of $[0.7, 0.95]$.

\subsection{Parameter Configuration Summary}
Table \ref{tab:sim_params_appendix} lists the simulation parameters, their specific value ranges used in the evaluation, and the corresponding data sources or theoretical justifications. For privacy protection, the discretization steps follow a proportional constraint $\Delta r_n \propto 1/\mathbbm{r}_n$ to maintain consistent entropy levels across different privacy radii.

\begin{table}[H]
	\centering
	\caption{Detailed Simulation Settings and Data Mapping}
	\label{tab:sim_params_appendix}
	\scriptsize 
	\begin{tabular}{llcl}
		\toprule
		\textbf{Parameter} & \textbf{Symbol} & \textbf{Value/Range} & \textbf{Data Source / Note} \\
		\midrule
		Max. buyers & $N_{Max}$ & 200 & DAIR-V2X trajectories \\
		Max. sellers & $M_{Max}$ & 50 & RCooper RSU density \\
		Time horizon & $T$ & 100 & Window from \cite{zhao2020social} \\
		Service types & $J$ & 5 & Available WAVE SCHs \\
		Grid size & - & $5000\text{m} \times 5000\text{m}$ & DAIR-V2X topology \\
		\midrule
		\textbf{Economic} & & & \\
		Buyer valuation & $v_n$ & [1, 10] & VCG calibrated \\
		Privacy cost & $k_n$ & [0.5, 1] & VCG calibrated \\
		Seller cost & $c_m$ & [1, 5] & RSU energy model \\
		\midrule
		\textbf{Privacy} & & & \\
		Privacy budget & $\xi_{n}^{t}$ & [2.5, 5] & $\xi_0=2.5$ (Initial) \\
		Privacy radius & $\mathbbm{r}_n$ & [3, 5] grids & Based on \cite{zhang2024trajectory} \\
		Radius step & $\Delta r_n$ & $0.5 \sim 1.0$ & $\propto 1/\mathbbm{r}_n$ \\
		Angle step & $\Delta \theta_n$ & $\pi/12 \sim \pi/6$ & Clustering from \cite{10815979} \\
		Similarity thresh. & $\delta_\Gamma$ & $0.85 \sim 0.95$ & Normalized Fréchet \\
		\midrule
		\textbf{Demand} & & & \\
		Quality demand & $\bm{\mathcal{Q}}_{n}^{t}$ & [0.7, 0.95] & HighD model \\
		Decay factor & - & $e^{-0.03t}$ & Time-varying logic \\
		\bottomrule
	\end{tabular}
\end{table}

\end{document}